\documentclass[twoside]{article}
\usepackage[accepted]{aistats2016}

\usepackage{amsmath,amssymb,amsthm,bm,color,graphicx,wasysym,placeins,caption,subcaption,url}
\newtheorem{theorem}{Theorem}
\newtheorem*{theorem*}{Theorem}

\newtheorem{corollary}{Corollary}
\newtheorem*{corollary*}{Corollary}
\newtheorem*{proposition*}{Proposition}

\begin{document}

%

%

\twocolumn[

\aistatstitle{Control Functionals for Quasi-Monte Carlo Integration}

\aistatsauthor{ Chris. J. Oates \And Mark Girolami }

\aistatsaddress{ University of Technology Sydney and \And University of Warwick and \\ The Alan Turing Institute for Data Science  } ]

\fontdimen14\textfont2=3pt

\begin{abstract} 
Quasi-Monte Carlo (QMC) methods are being adopted in statistical applications due to the increasingly challenging nature of numerical integrals that are now routinely encountered.
For integrands with $d$-dimensions and derivatives of order $\alpha$, an optimal QMC rule converges at a best-possible rate $O(N^{-\alpha/d})$.
However, in applications the value of $\alpha$ can be unknown and/or a rate-optimal QMC rule can be unavailable.
Standard practice is to employ $\alpha_L$-optimal QMC where the lower bound $\alpha_L \leq \alpha$ is known, but in general this does not exploit the full power of QMC.
One solution is to trade-off numerical integration with functional approximation.
This strategy is explored herein and shown to be well-suited to modern statistical computation.
A challenging application to robotic arm data demonstrates a substantial variance reduction in predictions for mechanical torques.
\end{abstract} 

\section{Introduction} \label{intro}

Consider a Lebesgue-integrable test function $f : \mathcal{X} \rightarrow \mathbb{R}$ defined on a bounded measurable subspace $\mathcal{X} \subseteq \mathbb{R}^d$ ($d \in \mathbb{N}$) with square integrable derivatives of order $\alpha>0$ in each variable.
Our focus is numerical computation of the integral $I[f] :=  \int_{\mathcal{X}} f(\bm{x}) d\bm{x}$.
The Quasi-Monte Carlo (QMC) approach is based on an approximation 
$$Q[f;\bm{x}^{1:N}] := \frac{1}{N} \sum_{n=1}^N f(\bm{x}^n)$$
where the (possibly random) design points $\bm{x}^{1:N} = \{\bm{x}^1,\dots,\bm{x}^N\} \subset \mathcal{X}$ have low discrepancy; that is, the points are `well-spaced' in a precise sense defined below.
This contrasts with the Monte Carlo (MC) approach whereby the design points are sampled independently from a uniform distribution over $\mathcal{X}$.
MC integration achieves a root mean square error (RMSE) convergence rate of $O(N^{-1/2})$ whereas QMC integration can in principle achieve a rate $O(N^{-\alpha/d})$ on specific geometric sequences $\{N_n\}_{n=1}^\infty$ \cite{Owen3}.
It is known that this rate is best-possible \cite{Novak} and explicit algorithms to generate design points that attain this rate are now available for many (but not all) values of $\alpha$ \cite{Dick6}.
Challenging integration problems are common in contemporary statistics, for example when computing expectations, marginal probability densities or normalising constants, and QMC methods are therefore gaining importance in statistical applications \cite{Gerber,Lacoste,Yang}.

Contrary to the above theoretical considerations, rate-optimal QMC is often not employed in practice.
This is mainly due to three reasons; either (R1) the smoothness parameter $\alpha$ is unknown, (R2) there does not currently exist an explicit QMC rule that is rate-optimal for functions of smoothness $\alpha$, or (R3) it is simply more convenient to employ a basic QMC rule based on a weaker smoothness assumption $\alpha_L < \alpha$, as implemented in standard software.
In each situation there is a gap between theory and practice that, as we show in this paper, can be bridged using functional approximation.

Previous work on variance reduction techniques for QMC includes \cite{Aistleitner}, who considered modified importance sampling strategies, and \cite{Hickernell}, who considered constructing control variates for QMC.
Neither approach improved the asymptotic error rate, though in some cases the QMC error was reduced by a constant factor.
Interestingly, \cite{Hickernell} reports some quite negative results for control variate strategies in this setting, because the objective being minimised by QMC is not equivalent to the MC variance that is minimised by control variates.
\cite{Wang2} demonstrates variance reductions in QMC are possible using additive approximations, though again the asymptotics were unchanged.

This paper studies a general approach to variance reduction for QMC rules, building on kernel methods and recent work in the Monte Carlo setting due to \cite{Oates,Tracey}.
The mathematics that underpins our work comes from the functional approximation literature.
This takes the form of a `control functional' $\psi : \mathcal{X} \rightarrow \mathbb{R}$ that satisfies (i) $\psi$ integrates to zero, (ii) $f - \psi$ is more amenable to QMC methods than $f$, in a precise sense.
The general approach that we explore is to replace the integrand $f$ by $f - \psi$ and target the QMC objective directly.
This can lead to accelerated asymptotics.
The main contribution of this paper is to explore this strategy in the settings (R1-3) above.
Theoretical analysis of convergence rates is provided, along with empirical results and a challenging application to robotics.
We begin by presenting some background on QMC theory below, before describing the methodology in more detail.

\section{Background}

QMC is naturally studied in reproducing kernel Hilbert spaces (RKHS; \cite{Dick5}). 
Below we draw connections with kernel methods, that are themselves naturally studied in RKHS.

{\bf Notation.} 
We work in a Hilbert space $H$, consisting of measurable functions $f : \mathcal{X} \rightarrow \mathbb{R}$.
For simplicity of presentation we assume $H$ includes the constant functions.
We follow the mainstream QMC literature by taking $\mathcal{X} = [0,1]^d$, equipped with the Euclidean norm $\|\bm{x}\| := (\sum_{i=1}^d x_i^2)^{1/2}$.
Denote the scalar product and norm on $H$ by $\langle\cdot,\cdot\rangle_{H}$ and $\|\cdot\|_{H}$ respectively.
Suppose further that $H$ is a RKHS with kernel $K : [0,1]^d \times [0,1]^d \rightarrow \mathbb{R}$; that is, $K$ satisfies (i) $K(\cdot,\bm{x}) \in H$ for all $\bm{x} \in [0,1]^d$ and (ii) $f(\bm{x}) = \langle f,K(\cdot,\bm{x})\rangle_{H}$ for all $f \in H$ and all $\bm{x} \in [0,1]^d$.
$K$ is assumed to be non-trivial, i.e. $K \neq 0$.

{\bf Quadrature Error Analysis.}
The quadrature methods that we focus on aim to minimise the `worst case' integration error which, for design points $\bm{x}^{1:N}$ and Hilbert space $H$, is defined to be
\begin{eqnarray}
e_H(\bm{x}^{1:N}) := \sup_{\|f\|_H \leq 1} \left| Q[f;\bm{x}^{1:N}] - I[f] \right|
\label{wce}
\end{eqnarray}
where the supremum is taken over all test functions $f$ belonging to the unit ball in $H$.
It follows from linearity that, for any function $f \in H$, the integration error obeys 
\begin{eqnarray}
\left| Q[f;\bm{x}^{1:N}] - I[f] \right| \leq e_H(\bm{x}^{1:N}) \|f\|_H. \label{KH}
\end{eqnarray}
The worst case error $e_H(\bm{x}^{1:N})$ is the usual target of QMC innovation, with $\bm{x}^{1:N}$ chosen to (approximately, asymptotically) minimise $e_H(\bm{x}^{1:N})$ \cite{Dick5}.
Note that Eqn. \ref{wce} is also the `maximum mean discrepancy' (MMD), as studied extensively in the kernel methods literature \cite{Briol,Smola}.

Quadrature is naturally studied in RKHS because there exists a closed-form expression for the worst case error in terms of the kernel $K$, which facilitates the principled selection of design points \cite{Dick5}:
\begin{eqnarray}
e_H(\bm{x}^{1:N})^2 & = & \int\int_{[0,1]^d} K(\bm{x},\bm{y}) d\bm{x} d\bm{y} \nonumber \\
& & - \frac{2}{N} \sum_{n=1}^N \int_{[0,1]^d} K(\bm{x}^n,\bm{y}) d\bm{y} \nonumber \\
& & + \frac{1}{N^2} \sum_{m,n=1}^N K(\bm{x}^n,\bm{x}^m)   \label{discrep}
\end{eqnarray}
The mainstream QMC literature supposes $H$ is a Sobolev space of known order $\alpha$ (defined below).
In this setting, $O(N^{-\alpha/d})$ is the best-possible rate for the worst case error when $\bm{x}^{1:N}$ are chosen deterministically and $O(N^{-\alpha/d- 1/2})$ is the best-possible RMSE when $\bm{x}^{1:N}$ are allowed to be random \cite{Novak}.
We will refer to QMC rules that achieve these optimal rates as `$\alpha$-QMC rules'.

This paper focuses on improving performance in the situation where a (sub-optimal) $\alpha_L$-QMC rule is used to integrate a test function of smoothness $\alpha > \alpha_L$.
For reasons (R1-3), this scenario is commonly encountered in statistical applications.
In contrast to QMC \cite{Dick5} (and kernel methods that aim to minimise the MMD \cite{Bach}), the rate constant $\|f\|_H$ is the primary target of our methodology below.

\section{Methodology}

{\bf Control Functionals for QMC.}
The approach that we pursue in this paper aims to construct a Lebesgue-integrable functional $\psi : [0,1]^d \rightarrow \mathbb{R}$ that satisfies 
\begin{eqnarray}
I[\psi] = 0. \label{zero}
\end{eqnarray}
When $\bm{x}$ has the interpretation of a random variable, $\psi(\bm{x})$ is classically known as a `control variate' \cite{Hickernell}.
When $\psi$ itself is estimated, we follow \cite{Oates} and refer to the entire mapping $\psi$ as a `control functional' (CF).
In the CF approach to estimation, the test function $f$ is replaced by $f - \psi$; it is hoped that the latter is more amenable to numerical integration.
Clearly $I[f - \psi] = I[f]$.
In this paper we construct a CF $\psi_N$ based on a tractable approximation $f_N$ to $f$.
(The dependence on $N$ will be explained below.)
It is required that the integral $I[f_N]$ is available in closed-form.
We then set
\begin{eqnarray}
\psi_N(\bm{x}) = f_N(\bm{x}) - I[f_N] \label{subtract}
\end{eqnarray}
so that $\psi_N$ satisfies Eqn. \ref{zero}.
For this to make sense mathematically, it must be the case that $f_N \in H$ and this informs our method of approximation (the constant function with value $I[f_N]$ belongs to $H$ by assumption).
Intuitively, a good CF $\psi_N$ will provide a close approximation to fluctuations of the test function $f$, so that the functional difference $f - \psi_N$ become increasingly `flat' and thus more amenable to QMC methods.
More precisely, motivated by Eqn. \ref{KH} we aim to construct a CF such that $\|f - \psi_N\|_H < \|f\|_H$.
This connection with functional approximation offers the possibility to leverage kernel methods for these problems, see e.g. \cite{Schaback,Tracey}.

{\bf Control Functional Error Analysis.}
Consider partitioning $\bm{x}^{1:N}$ into two sets $\bm{u}^{1:M}$ and $\bm{v}^{M+1:N}$ where $1 < M < N$ and $M/N \rightarrow c \in (0,1)$ as $N \rightarrow \infty$.
The first set $\bm{u}^{1:M}$, possibly non-random, will be used in a preliminary step to construct an approximation $f_M(\cdot;\bm{u}^{1:M})$ to $f$.
Then the second set $\bm{v}^{M+1:N}$, possibly random, is used to evaluate the `CF estimator'
\begin{eqnarray}
E[f;\bm{u}^{1:M},\bm{v}^{M+1:N}] & := & Q[f-\psi_N(\cdot;\bm{u}^{1:M});\bm{v}^{M+1:N}] \nonumber \\ 
& = &  Q[f-f_M(\cdot;\bm{u}^{1:M});\bm{v}^{M+1:N}] \nonumber \\
& & \; \; \; \; \; \; \; \;  + I[f_M(\cdot;\bm{u}^{1:M})]. 
\end{eqnarray}
We remark that if the points $\bm{v}^n$ are random and marginally distributed as $U([0,1]^d)$ then $E[f;\bm{u}^{1:M},\bm{v}^{M+1:N}]$ will be an unbiased estimator for $I[f]$.
Error analysis for the CF estimator is based on the following:
\begin{theorem} \label{theorem one}
Given $f,f_M \in H$, we have
\begin{align}
& | E[f;\bm{u}^{1:M},\bm{v}^{M+1:N}] - I[f] | \nonumber \\
& \; \; \; \; \; \; \; \; \; \; \; \; \; \; \; \leq  e_H(\bm{v}^{M+1:N}) \|f-f_M(\cdot;\bm{u}^{1:M})\|_H. \label{QMCCF}
\end{align}
\end{theorem}
\begin{proof}
Since $f,f_M \in H$ we have that $f-f_M \in H$. 
The result then follows by applying the fundamental inequality from Eqn. \ref{KH} to the function $f-f_M$ and using linearity of the integral operator $I$.
\end{proof}
Thus the CF methodology produces an estimator $E[f;\bm{u}^{1:M},\bm{v}^{M+1:N}]$ that has asymptotically zero error relative to standard QMC estimators, providing that it is possible to construct an approximation $f_M$ to $f$ in such a way that $\|f - f_M(\cdot;\bm{u}^{1:M})\|_H \rightarrow 0$ as $M \rightarrow \infty$.
The next sections establish convergence rates for functional approximation using kernel methods.

{\bf Sobolev Spaces.}
To achieve consistent approximation $\|f-f_M\|_H \rightarrow 0$ it is necessary to impose regularity conditions on $H$. 
Sobolev spaces are a general setting in which to formulate such regularity assumptions; our main reference here is \cite{Schaback}.
Firstly suppose that $k \in \mathbb{N}_0$, $k > d/2$ and $1 \leq p < \infty$. 
For a multi-index $\bm{a} \in \mathbb{N}_0^d$ we write $|\bm{a}| = a_1+\dots+a_d$.
Define the `$p$-Sobolev space of order $k$' to be
\begin{eqnarray*}
W^{k,p} & := & \{f : [0,1]^d \rightarrow \mathbb{R} \; | \;  D^{\bm{a}}f \text{ exists and } \\
& & \; \; D^{\bm{a}}f \in L_p([0,1]^d), \forall \bm{a} \in \mathbb{N}_0^d \text{ with } |\bm{a}| \leq k\}.
\end{eqnarray*}
Here $D^{\bm{a}} f$ denotes the weak (or `distributional') derivative of $f$; the reader is referred to the above reference for details.
Clearly $W^{k,p}$ is a vector space over $\mathbb{R}$ when addition and (scalar) multiplication are defined point-wise.
For the special case $p=2$ we equip $W^{k,2}$ with the inner product
\begin{eqnarray*}
\langle f , g \rangle_k := \sum_{\bm{a} \in \mathbb{N}_0^d, |\bm{a}| \leq k} I[D^{\bm{a}} f D^{\bm{a}} g]
\end{eqnarray*}
and denote this inner-product space $H^{k} := (W^{k,2},\langle\cdot,\cdot\rangle_k)$.
Defined in this way, $H^k$ is a Hilbert space of functions whose (weak) derivatives exist up to order $k$.
Moreover $H^k$ can be made into a RKHS with an appropriate choice of kernel (see below).
Our results below apply also to Sobolev spaces with non-integer $k$, but this construction is more technical and we refer the reader to \cite{Schaback} for details.

{\bf Approximation in Sobolev Spaces.}
Our assumptions are naturally stated using Sobolev spaces:
Given two Hilbert spaces $H$, $H'$, defined on the same element set, with norms $\|\cdot\|_H$, $\|\cdot\|_{H'}$, we say that $H$ and $H'$ are `norm-equivalent', written  $H \equiv H'$, whenever there exist positive constants $c_1$, $c_2$ such that $c_1\|f\|_H \leq \|f\|_{H'} \leq c_2\|f\|_H$ for all $f \in H$.

\noindent {\it Assumption 1}: $H \equiv H^{\alpha_L}$ where $\alpha_L > d/2$.

\noindent {\it Assumption 2}: $f \in H^\alpha$ where $\alpha \geq \alpha_L$.

Assumption 1 is a technical requirement to ensure the space $H$ (where QMC is performed) admits consistent functional approximation.
Assumption 2 ensures that the test function $f$ is `smooth enough' for $\alpha_L$-QMC methods to converge at the $\alpha_L$-rate.
This follows from the fact that Sobolev spaces are nested, so that $f \in H^\alpha \implies f \in H^{\alpha_L}$.

For consistent approximation of $f$ it is necessary to base our approximation $f_M$ in a space $H_*$ of functions that are `at least as smooth' as $f$: 

\noindent {\it Assumption 3}: $H_* \equiv H^{\alpha_U}$ where $\alpha_U \geq \alpha$.

It follows again from the nested property that $f_M \in H^{\alpha_L}$ and thus the functional difference $f - f_M$ exists in $H^{\alpha_L}$.
The Sobolev spaces $H_*$ can be characterised as RKHS via an appropriate reproducing kernel $K_*$, such as the well-known Mat\'{e}rn kernel.

Finally an approximation $f_M$ to $f$ is constructed based on the points $\bm{u}^{1:M}$ as follows:
\begin{eqnarray}
f_M(\bm{x};\bm{u}^{1:M}) := \sum_{n=1}^M \beta_n K_*(\bm{x},\bm{u}^n) \label{reg 1}
\end{eqnarray}
where the weights $\beta_n \in \mathbb{R}$ are defined as the solution to the linear system of interpolation equations
\begin{eqnarray}
f_M(\bm{u}^n;\bm{u}^{1:M}) = f(\bm{u}^n), \; \; \; \; \; n = 1,\dots,M. \label{interp}
\end{eqnarray}
It is well-known that Eqn. \ref{reg 1} is the unique minimiser of the $H_*$-norm under all functions in $H_*$ that satisfy the linear system in Eqn. \ref{interp} \cite{Schaback}.
In practice it may be necessary to regularise the linear system in order to facilitate inversion, but we do not go into details here, see e.g. \cite{Schaback}.

We note that $I[f_M]$ will {\it not} have a closed-form expression when the Mat\'{e}rn kernel is employed and for this technical reason we instead employ tensor products of polynomial kernels (these  give rise to Sobolev spaces of mixed dominating smoothness - full details are provided at the end of this section).

{\bf Theory: Deterministic Case.}
We begin by considering the case where the design points $\bm{v}^{M+1:N}$ are chosen deterministically.
Define the `fill distance' 
$$h(\bm{u}^{1:M}) := \sup_{\bm{x} \in [0,1]^d} \min_{n} \|\bm{x} - \bm{u}^n\|,$$ 
the `separation radius' 
$$q(\bm{u}^{1:M}) := \frac{1}{2} \min_{j \neq k}\|\bm{u}^j - \bm{u}^k\|$$ 
and the `mesh ratio' $\rho(\bm{u}^{1:M}) := h(\bm{u}^{1:M}) / q(\bm{u}^{1:M})$.
The set $\bm{u}^{1:M}$ is called `quasi-uniform' if $\rho(\bm{u}^{1:M}) \rightarrow 1$ as $M \rightarrow \infty$.
\begin{theorem} \label{theo1}
Under Assumptions 1-3 the CF estimator has error bounded by
\begin{align*}
& | E[f;\bm{u}^{1:M},\bm{v}^{M+1:N}] - I[f] |  \\
& \; \;   \leq C e_{H^{\alpha_L}}(\bm{v}^{M+1:N}) h(\bm{u}^{1:M})^{\alpha - \alpha_L} \rho(\bm{u}^{1:M})^{\alpha_U-\alpha_L} \|f\|_{H^\alpha}
\end{align*}
where $C>0$ is a constant that depends on $\alpha$, $\alpha_L$ and $\alpha_U$ but not on $f$, $\bm{v}^{M+1:N}$ and $\bm{u}^{1:M}$.
\end{theorem}
\begin{proof}
From \cite{Schaback} (Theorem 7.8) we have that the kernel estimator in Eqn. \ref{reg 1} is consistent for the non-parametric regression problem at a rate
\begin{align*}
& \|f-f_M(\cdot;\bm{u}^{1:M})\|_{H^{\alpha_L}} \\
& \; \; \; \; \; \; \; \; \; \; \leq C h(\bm{u}^{1:M})^{\alpha - \alpha_L} \rho(\bm{u}^{1:M})^{\alpha_U-\alpha_L} \|f\|_{H^\alpha}
\end{align*}
where $C$ depends only on $\alpha, \alpha_L, \alpha_U$.
Combining this with Eqn. \ref{QMCCF} completes the proof.
\end{proof}

For quasi-uniform $\bm{u}^{1:M}$, there is no asymptotic penalty from employing a kernel $K_*$ that imposes `too much smoothness' on the approximation $f_M$, with $\rho \rightarrow 1$.
In this case $h(\bm{u}^{1:M}) = O(M^{-1/d})$ and, since $M$ and $N$ are proportional, $h(\bm{u}^{1:M}) = O(N^{-1/d})$.
However the rate constant $C$ will increase when too much smoothness is assumed so that, as a rule of thumb, we should try to select $\alpha_U$ as close as possible to $\alpha$.
Our main result is stated below:

\begin{corollary}
When $\bm{u}^{1:M}$ is quasi-uniform, CFs accelerate $\alpha_L$-QMC by a factor $O(N^{- (\alpha - \alpha_L)/d})$.
\end{corollary}

{\it Remark:} The improvement due to CFs appears to be mainly limited to low-dimensional integrals ($d$ small), but in fact CFs can in principle be extended to high-dimensional integrals under additional tractability assumptions, as discussed in Sec. \ref{discuss}.

{\it Remark:} Optimising the bound in Theorem \ref{theo1} enables us to obtain the optimal scaling
$$
\frac{M}{N} \rightarrow c^* = \frac{\alpha - \alpha_L}{\alpha},
$$
see the Supplement for full details.

The overall convergence rate of the CF estimator depends on how the design points $\bm{v}^{M+1:N}$ are generated. 
For this there are many QMC methodologies available, each leading to different convergence rates for the worst case error $e_{H^{\alpha_L}}(\bm{v}^{M+1:N})$; see \cite{Dick} for a recent survey of some of these approaches.
Of particular interest in statistical applications is the case of random design points which we discuss below.

{\bf Theory: Randomised Case.} 
Modern QMC methods begin with a deterministic set/sequence of design points (e.g. a Halton sequence or a Sobol sequence), then apply a random transformation leading to a low discrepancy set with high probability.
Below we consider three types of randomisation; shifting, folding and scrambling.

{\it Shifting:}
In `random shift' QMC the design points $\bm{v}^{M+1:N}$ are translated by a common uniform random vector $\bm{\Delta} \in [0,1]^d$, so that $\bm{v}^n \mapsto \bm{v}^n + \bm{\Delta}$ for each $n = M+1,\dots,N$. For convenience we write this `shifted' set as $\bm{v}^{M+1:N} + \bm{\Delta}$.
Applying Theorem \ref{theo1} to $\bm{v}^{M+1:N} + \bm{\Delta}$ and then marginalising over $\bm{\Delta} \in [0,1]^d$ produces a RMSE bound for the CF estimator:
\begin{corollary}
Under Assumptions 1-3 the random shift CF estimator has error bounded by
\begin{align*}
& \sqrt{\mathbb{E}|E[f;\bm{u}^{1:M},\bm{v}^{M+1:N}+\bm{\Delta}] - I[f]|^2} \\
& \; \; \leq C e_{H^{\alpha_L}}^{\text{sh}}(\bm{v}^{M+1:N}) h(\bm{u}^{1:M})^{\alpha - \alpha_L} \rho(\bm{u}^{1:M})^{\alpha_U-\alpha_L} \|f\|_{H^\alpha}
\end{align*}
where 
\begin{eqnarray*}
(e_{H^{\alpha_L}}^{\text{sh}}(\bm{v}^{M+1:N}))^2 := \int_{[0,1]^d} e_{H^{\alpha_L}}(\bm{v}^{M+1:N}+\bm{\Delta})^2 d\bm{\Delta}
\end{eqnarray*}
and $C>0$ is a constant that does not depend on $f$, $\bm{v}^{M+1:N}$ or $\bm{u}^{1:M}$.
\end{corollary}
For quasi-uniform $\bm{u}^{1:M}$, CFs accelerate random shift $\alpha_L$-QMC by a factor $O(N^{- (\alpha - \alpha_L)/d})$ (compare against Sec. 5.2 of \cite{Dick}).

{\it Folding:}
A shifted and `folded' QMC rule takes the form 
$$Q_{\bm{b}}(f;\bm{z}^{1:N}+\bm{\Delta}) := \frac{1}{N} \sum_{n=1}^N f(\bm{b}(\bm{z}^n + \bm{\Delta}))$$ 
where $\bm{b}$ is the `baker's transformation', given by $b_i(\bm{t}) = 1 - |2t_i-1|$.
This transformation reduces error rates; for example, for $f \in SH^2([0,1]^d)$ (defined below), folding and shifting a uniform lattice $\bm{z}^{1:N}$ leads to a RMSE $O(N^{-2+\epsilon})$ that is smaller than the RMSE $O(N^{-1+\epsilon})$ for a shifted lattice (p. 59 of \cite{Dick}).
The CF estimator here is  
\begin{align*}
& E_{\bm{b}}[f;\bm{u}^{1:M},\bm{v}^{M+1:N}+\bm{\Delta}] \\
& \; \; := I[f_M(\cdot;\bm{u}^{1:M})] + Q_{\bm{b}}[f-f_M(\cdot;\bm{u}^{1:M});\bm{v}^{M+1:N}+\bm{\Delta}].
\end{align*}
For convenience we denote the shifted and folded design points by $\bm{b}(\bm{v}^{M+1:N} + \bm{\Delta})$.
Applying Theorem \ref{theo1} to $\bm{b}(\bm{v}^{M+1:N} + \bm{\Delta})$ and then marginalising over $\bm{\Delta} \in [0,1]^d$ produces:
\begin{corollary}
Under Assumptions 1-3 the shifted and folded CF estimator has error bounded by
\begin{align*}
& \sqrt{\mathbb{E}|E_{\bm{b}}[f;\bm{u}^{1:M},\bm{v}^{M+1:N}+\bm{\Delta}] - I[f]|^2}  \\
& \; \;  \leq C e_{H^{\alpha_L}}^{\text{sh},\bm{b}}(\bm{v}^{M+1:N}) h(\bm{u}^{1:M})^{\alpha - \alpha_L} \rho(\bm{u}^{1:M})^{\alpha_U-\alpha_L} \|f\|_{H^\alpha}
\end{align*}
where 
\begin{eqnarray*}
(e_{H^{\alpha_L}}^{\text{sh},\bm{b}}(\bm{v}^{M+1:N}))^2 := \int_{[0,1]^d} e_{H^{\alpha_L}}(\bm{b}(\bm{v}^{M+1:N}+\bm{\Delta}))^2 d\bm{\Delta}
\end{eqnarray*}
and $C>0$ is a constant independent of $f$, $\bm{v}^{M+1:N}$ and $\bm{u}^{1:M}$.
\end{corollary}
Again, for quasi-uniform $\bm{u}^{1:M}$, CFs accelerate shifted and folded $\alpha_L$-QMC by a factor $O(N^{- (\alpha - \alpha_L)/d})$ (compare against Sec. 5.9 of \cite{Dick}).

{\it Scrambling:}
An explicit $\alpha$-QMC rule that applies for all integer values of $\alpha$ was recently discovered by \cite{Dick6}.
For simplicity focussing on $d=1$, these random design points achieve $\alpha$-rates and, moreover, the RMSE is controlled by a norm of the form $\|f\|_{H^{\alpha}}$.
When $\alpha$ is known and is an integer, one may achieve optimal rates and CFs provide no rate improvement.
However, when $\alpha \notin \mathbb{N}$, CFs can be used to transform these sub-optimal integrators into optimal integrators.

{\bf Choice of Kernel:}
The QMC+CF methodology has some flexibility in terms of the choice of kernel $K_*$ that is used to construct the approximation $f_M$.
Our main requirements here are:
(i) $K_*$ imposes `enough smoothness' on $f_M$ in order to be able to faithfully approximate $f$ (Assumption 3).
Moreover, $K_*$ should be tunable to achieve a pre-specified minimum level of smoothness.
Below we make an explicit connection between $K_*$ and the order of the associated `native' Sobolev space that will allow us to satisfy this requirement.
(ii) The functions $K_*(\cdot,\bm{y})$ can be integrated analytically, so that $I[f_M]$ is available in closed form.
This second requirement leads us to consider tensor products of Sobolev spaces, as described below.

To construct analytically integrable functional approximations we consider kernels that are given by polynomials.
Wendland's compactly supported functions \cite{Wendland} are defined via the recursion 
$$\varphi_{d,k} = \mathcal{I}^k [\varphi_{\lfloor d/2 \rfloor+k+1}],$$ 
the base function $\varphi_\ell(r) = (1-r)_+^\ell$ with $x_+ := \max\{0,x\}$, and the integral operator 
$$\mathcal{I} [\varphi](r) = \int_r^\infty t \varphi(t)dt$$
($r \geq 0$), so that
\begin{eqnarray*}
\varphi_{d,k}(r) = \left\{\begin{array}{ll} (1-r)^{\ell+k} p_{d,k}(r), & r \in [0,1] \\ 0, & r > 1 \end{array} \right.
\end{eqnarray*}
where $\ell = \lfloor d/2 \rfloor + k + 1$ and $p_{d,k}$ is a polynomial of degree $k$ (see e.g. p.87 of \cite{Fasshauer} for explicit formulae).
Then the kernel $K_*(\bm{x},\bm{y}) = \varphi_{d,k}(\|\bm{x} - \bm{y}\|)$ has native space $H^{d/2+k+1/2}$ (where the restriction $d > 3$ is in principle required for the special case $k = 0$) (see e.g. p.109 of \cite{Fasshauer}).
With this kernel we can therefore guarantee a minimum level of smoothness.
By rescaling, the kernel's support can be changed from the unit ball (as above) to balls of smaller radius. 
This in turn enforces sparsity on the system of interpolation equations that are the basis of the CF estimator and reduces the computational cost of inverting this linear system.

Wendland's kernel cannot be integrated analytically in $d \geq 2$ dimensions, violating requirement (ii).
However we can exploit recent work by \cite{Sickel} that shows the $d$-dimensional tensor product space $H^k([0,1]) \otimes \dots \otimes H^k([0,1])$ 
is norm-equivalent to $SH^k = SH^k([0,1]^d)$, the Sobolev space with dominating mixed smoothness: 
\begin{eqnarray*}
SH^k & := & \{f : [0,1]^d \rightarrow \mathbb{R} \; | \; D^{\bm{a}}f \text{ exists and } \\
& & \; \; D^{\bm{a}}f \in L_p([0,1]^d), \forall \bm{a} \in \mathbb{N}_0^d \text{ with } a_i \leq k\}.
\end{eqnarray*}
(The distinction with $H^k([0,1]^d)$ is that the multi-index $\bm{a}$ is now constrained component-wise, $a_i \leq k$, rather than $|\bm{a}| \leq k$.)
In particular $SH^k([0,1]^d) \subseteq H^k([0,1]^d)$ so that functions in $SH^k$ are at least as smooth as functions in $H^k$.
We therefore propose to employ the product kernel
\begin{eqnarray}
K_*^{(k)}(\bm{x},\bm{y}) = \prod_{i=1}^d \varphi_{1,k}(|x_i-y_i|) \label{new kernel}
\end{eqnarray}
whose native space is $SH^{k + 1}$.
The integral 
$$
\int_{[0,1]^d} K_*^{(k)}(\bm{x},\bm{y}) d\bm{x}
$$ 
of tensor products of Wendland functions in Eqn. \ref{new kernel} can now be integrated analytically.
This approach provides a convenient mechanism to control the degree of smoothness that we impose on the approximation $f_M$.

\begin{figure*}[t!]
\centering
\includegraphics[width = 0.32\textwidth]{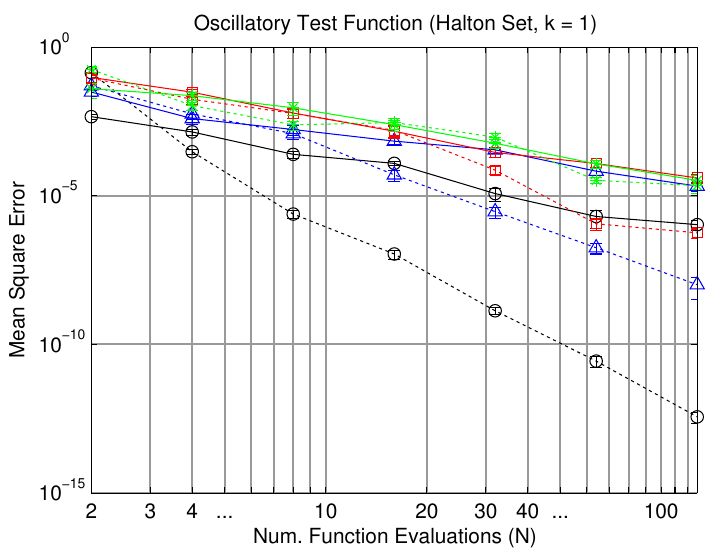}
\includegraphics[width = 0.32\textwidth]{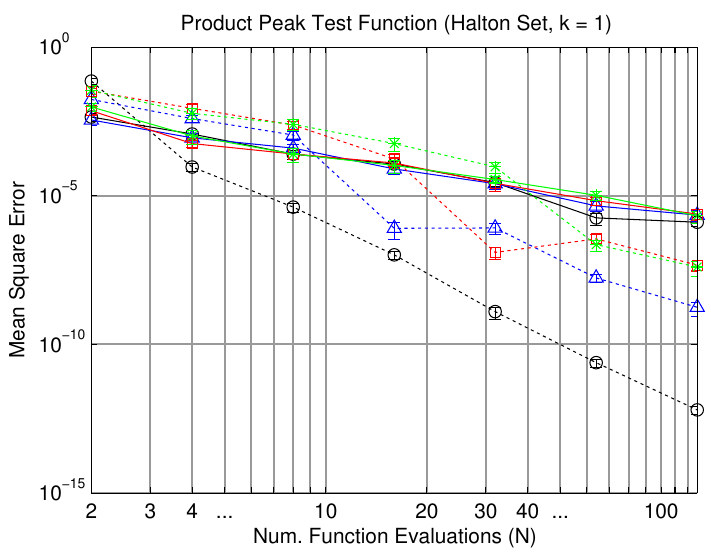}
\includegraphics[width = 0.32\textwidth]{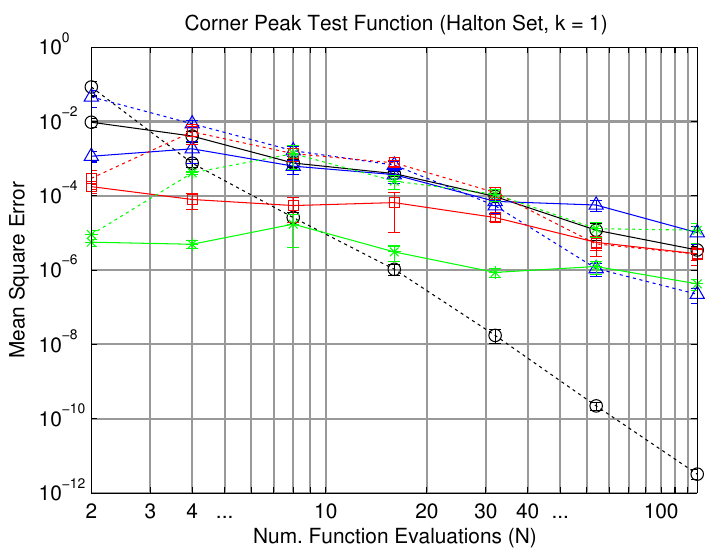}
\includegraphics[width = 0.32\textwidth]{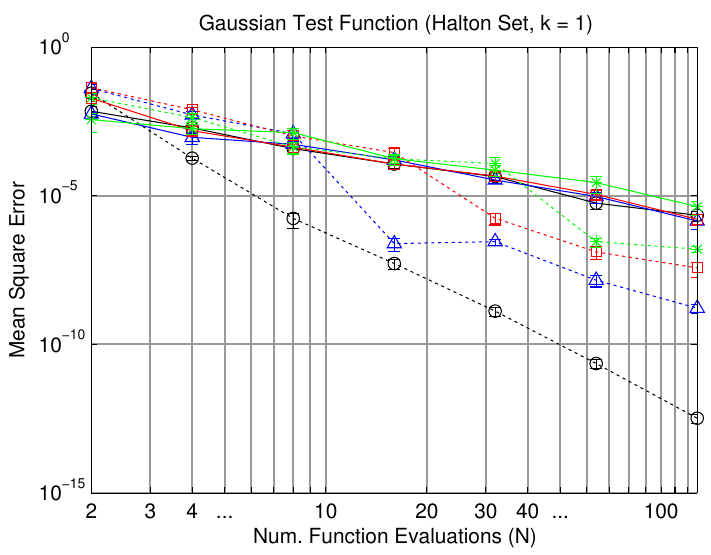}
\includegraphics[width = 0.32\textwidth]{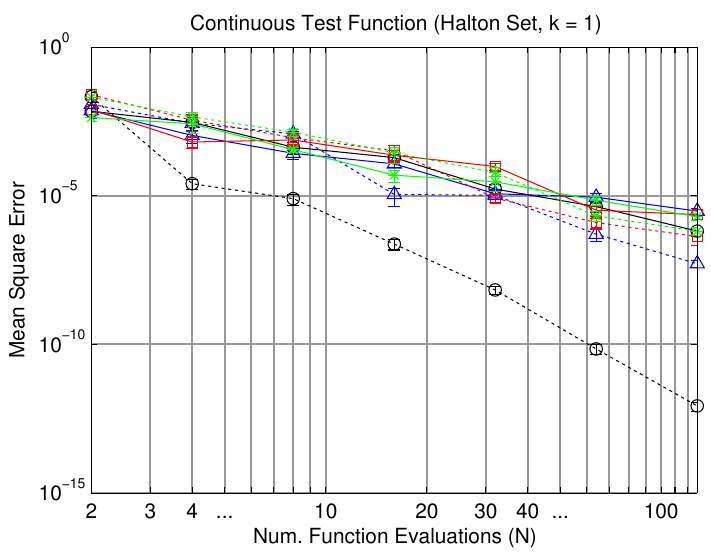}
\includegraphics[width = 0.32\textwidth]{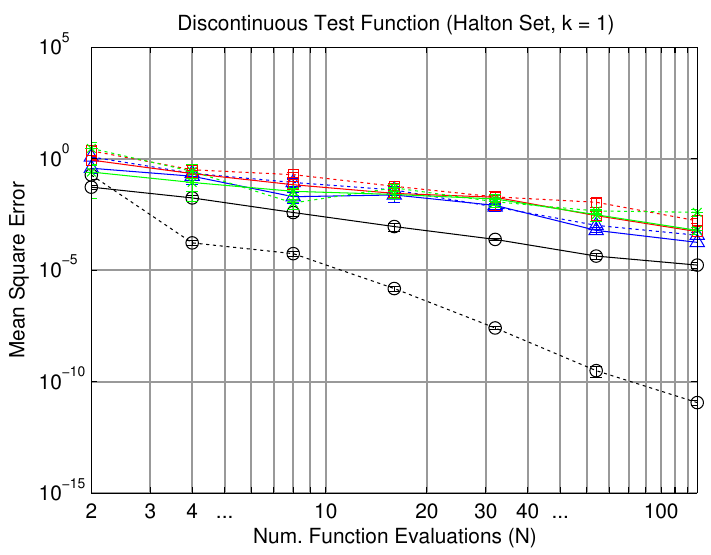}
\caption{Simulation study (Genz functions): 
Each panel represents one test function. 
Solid lines correspond to standard QMC, dashed lines correspond to QMC+CF.
$\ocircle$ represents dimension $d=1$, \textcolor{blue}{$\triangle$} represents $d=2$, \textcolor{red}{$\square$} represents $d = 3$ and \textcolor{green}{$*$} represents $d = 4$.
Experiments were replicated with 10 random seeds and error bars denote standard error of the replicate mean.
QMC points were generated from a shifted and scrambled Halton sequence.
A Wendland regression kernel was used with $k=1$.
}
\label{results fig}
\end{figure*}

\section{Experimental Results} \label{sims}

Our methodology provides a variance reduction technique for QMC that is able to accelerate convergence rates, yet is also practical.
The first numerical study below is a `proof-of-principle' designed to validate this specific claim in the empirical setting.

{\bf Simulation Study:}
For objective assessment we exploited the test package proposed by \cite{Genz}. 
This package defines 6 function families, each of them characterized by some peculiarity, such as oscillation, discontinuity or corner peaks, with the property that their exact integrals are available.
The `discontinuous' Genz function provides an example where smoothness assumptions on the test function are violated.
We used the MATLAB implementation of \cite{Genz} that is freely available at \url{http://people.sc.fsu.edu/~jburkardt/m_src/testpack/testpack.html}.

In the experiments below, we focus on the two QMC rules that are most widely used in practice.
In the first experiment, the random QMC point set $\bm{v}^{M+1:N}$ was generated by truncating the Halton sequence, scrambling the digits of the resulting points using the reverse-radix algorithm \cite{Kocis} and applying a uniform random shift. 
This QMC rule achieves the $\alpha_L = 1$ rate on the subsequence $N_n = 2^n$ when the test function has mixed partial derivatives of first order.
To ensure that these QMC rules were implemented faithfully, we restricted attention to the case where $M = N/2$ so that $N-M$ was always a power of two.
The training points $\bm{u}^{1:M}$ were taken to be $d$-dimensional square lattices in all experiments.

We considered the 6 Genz functions in $d = 1,2,3,4$ dimensions.
The performance of QMC with and without CFs was compared, in each case ensuring that the total number of evaluations of the integrand $f$ was equal for all methods.
For CFs, the tensor-product Wendland kernel with $k=1$ was employed (i.e. approximation with functions $f_M \in H^2$, so $\alpha_U = 2$).
Results are presented in Fig. \ref{results fig}.
(For clarity we chose not present results for MC, since these were inferior to QMC methods in all cases considered.)
For the first 5 Genz functions it holds that $f \in H^\alpha$ with $\alpha = 2$ and theory (for the random case) guarantees an acceleration of $O(N^{-1/d})$; this is borne out in experimental results.
In the 6th, discontinuous case the QMC+CF method does not out-perform QMC (at least in dimension $d>1$), as the functional approximation $f_M$ is poor due to violation of our continuity assumption.
In all cases the performance of QMC+CF approaches that of QMC as the dimension $d$ increased.
In higher dimensions ($d \geq 5$, not shown) the QMC+CF and QMC estimators demonstrated effectively identical performance, in line with theory.

The experiments were then repeated with rougher ($k=0$) and smoother ($k = 2$) regression kernels.
Results in the Supplement (Figs. S3-8) demonstrated a slight improvement in the performance of QMC+CF when $k=2$, in line with theory, though generally estimates were robust to the choice of regression kernel.
To further assess the generality of these conclusions, further experiments were performed using a different QMC rule (truncated Sobol sequence with scrambling due to \cite{Matousek}). 
Results in the Supplement showed that the same conclusions can be drawn in each case.
Taken together, these results demonstrate that CFs can accelerate QMC, at least in low-dimensional settings, and thus complete our `proof-of-principle'.
MATLAB code to reproduce these results is provided.

{\bf Application to Robot Arm Data:}
To demonstrate the benefits of our methodology we consider the problem of estimating the inverse dynamics of a seven degrees-of-freedom robot arm. 
The task, as described in \cite{Rasmussen}, is to map from a 21-dimensional input space (7 positions, 7 velocities, 7 accelerations) to the corresponding 7 joint torques. 
Following \cite{Rasmussen} we present results below on just one of the mappings, from the 21 input variables to the first of the seven torques.
The dataset consists of 48,933 input-output pairs, of which $44,484$ were used as a training set and the remaining 4,449 were used as a test set. The inputs were linearly rescaled to have mean zero and unit variance on the training set. The outputs were centred to have mean zero on the training set.

We consider a hierarchical model based on 21-dimensional Gaussian process (GP) regression.
Denote by $Y_i \in \mathbb{R}$ a measured response variable at state $\bm{z}_i \in \mathbb{R}^{21}$, assumed to satisfy $Y_i = g(\bm{z}_i) + \epsilon_i$ where $\epsilon_i \sim N(0,\sigma^2)$ are independent for $i = 1,\dots,n$ and $\sigma > 0$ will be assumed known.
In order to use training data $(y_i,\bm{z}_i)_{i=1}^n$ to make predictions regarding an unseen test point $\bm{z}_*$, we place a GP prior $g \sim \mathcal{GP}(0,c(\bm{z},\bm{z}';\bm{\theta}))$ where $c(\bm{z},\bm{z}';\bm{\theta}) = \theta_1 \exp (- \frac{1}{2} \theta_2^{-2} \|\bm{z} - \bm{z}'\|_2^2 )$.
Here $\bm{\theta} = (\theta_1,\theta_2)$ are hyper-parameters that control how training samples are used to predict the response at a new test point.
A fully-Bayesian treatment aims to marginalise over these hyper-parameters and we assign independent priors $\theta_1 \sim \Gamma(\alpha,\beta)$, $\theta_2 \sim \Gamma(\gamma,\delta)$ in the shape/scale parametrisation, which we write jointly as $\pi(\bm{\theta})$.
Here  $\sigma = 0.1$, $\alpha = \beta = \gamma = \delta = 2$.

To predict the value of the response $Y_*$ corresponding to an unseen state vector $\bm{z}_*$, our estimator will be the Bayesian posterior mean 
\begin{eqnarray}
\hat{Y}_* := \mathbb{E}[Y_*|\bm{y}] = \int \mathbb{E}[Y_*|\bm{y},\bm{\theta}] \pi(\bm{\theta}) d\bm{\theta}, \label{robot target}
\end{eqnarray}
where we implicitly condition on the covariates $\bm{z}_1,\dots,\bm{z}_n,\bm{z}_*$.
Phrasing in terms of our earlier notation, the test function is 
\begin{eqnarray*}
f(\bm{x}) = \mathbb{E}[Y_*|\bm{y},\Pi^{-1}(\bm{x})] = \bm{C}_{*,n} (\bm{C}_n + \sigma^2 \bm{I}_{n \times n})^{-1} \bm{y}
\end{eqnarray*}
where $\Pi$ is the c.d.f for $\pi$, $(\bm{C}_n)_{i,j} = c(\bm{z}_i,\bm{z}_j;\bm{\theta})$ and $(\bm{C}_{*,n})_{1,j} = c(\bm{z}_*,\bm{z}_j;\bm{\theta})$.
Each evaluation of the integrand $f(\bm{x})$ requires $O(n^3)$ operations due to the matrix inversion and this entails a prohibitive level of computation.
A partial solution is provided by a `subset of regressors' approximation
\begin{eqnarray}
f(\bm{x}) \approx \bm{C}_{*,n'} (\bm{C}_{n',n} \bm{C}_{n,n'} + \sigma^2 \bm{C}_{n'})^{-1} \bm{C}_{n',n} \bm{y} \label{sor approx}
\end{eqnarray}
where $n' < n$ denotes a subset of the full data; see Sec. 8.3.1 of \cite{Rasmussen} for full details.
However even Eqn. \ref{sor approx} still represents a substantial computational burden in general.
To facilitate the illustration below, which investigates the sampling distribution of estimators, we took a random subset of $n=1,000$ training points and a subset of regressors approximation with $n' = 100$.
The total computational time needed to obtain these results was 268 core-hours.

\begin{figure*}[t!]
\centering
\includegraphics[width = 0.3\textwidth,clip,trim = 0cm 1cm 0cm 0cm]{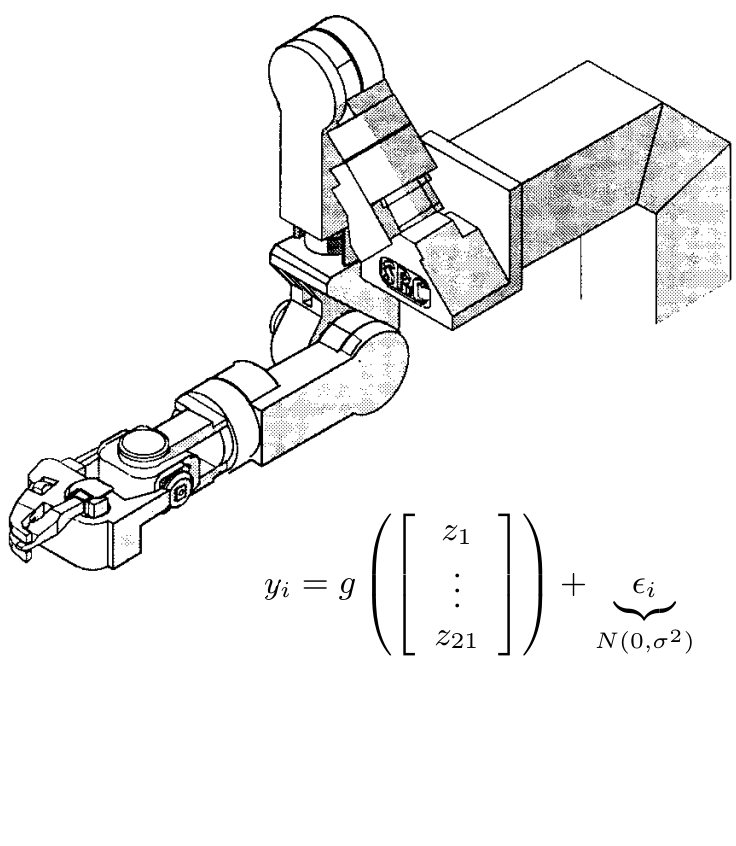}
\includegraphics[width = 0.6\textwidth,clip,trim = 6cm 0.2cm 0cm 0cm]{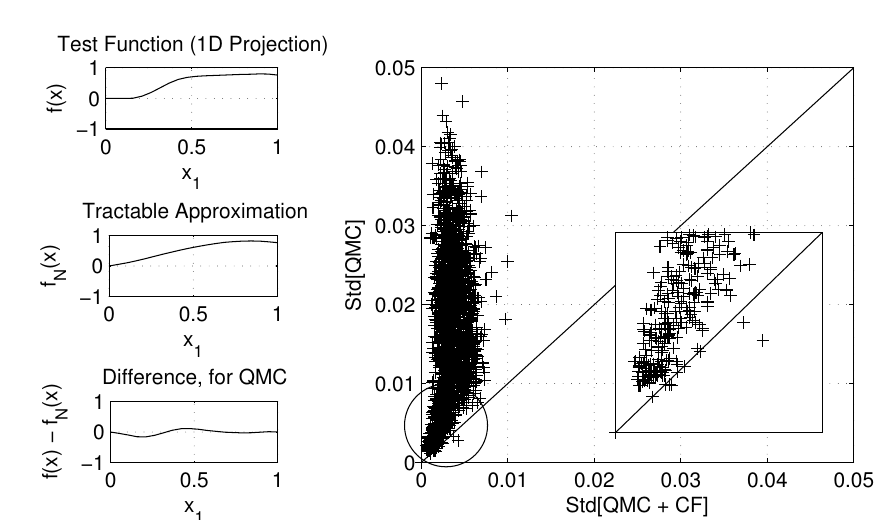}
\caption{Application to robot arm data.
{\it Left}: Posterior predictive means were computed for the mechanical torque experienced by one of the seven joints of the arm, for each of 4,449 joint configurations.
Schematic reproduced from \cite{Vijayakumar2000}.
{\it Centre}: Model hyper-parameters were integrated out; for this task we compared standard QMC with the proposed QMC+CF approach (both implementations provided unbiased estimators).
{\it Right}: Examining the estimator sampling standard deviations, we see that, for all but a handful of the configurations, QMC+CF was more accurate than QMC.
}
\label{robot fig}
\end{figure*}

For each test point $\bm{z}_*$ the sampling standard deviation of $\hat{Y}_*$ was estimated from 10 independent realisations of the QMC procedures.
For CF we used a randomly-shifted, scrambled Halton sequence ($\alpha_L = 1$) and Wendland kernels with $k=1$ ($\alpha_U = 2$), so that theory predicts an acceleration factor of $O(N^{-1/2})$.
The estimator standard deviations were estimated for all 4,449 test points (with $N = 2^8$) and the full results are shown in Fig. \ref{robot fig}.
Note that each test point $\bm{z}_*$ corresponds to a different test function $f$ and thus these results are quite objective, encompassing thousands of different integration problems.
For the vast majority of integration problems, CF accelerated the standard QMC estimator.
Here the computational time to construct a functional approximation (inverting a $16 \times 16$ matrix) was negligible (3\%) in comparison to the cost of evaluating the function $f$ once.
The total additional computational time associated with the QMC+CF methodology was 2\% greater than for QMC, which is easily justified by the substantial variance reductions ($\sim 10^3\%$) that are realised in this application.
Supplementary results (Fig. S9) compare QMC+CF to MC+CF (standard MC sampling).

\section{Discussion} \label{discuss}

QMC methods are becoming increasingly relevant in modern statistics applications \cite{Gerber,Yang} and it is surely a priority to target the rate constants governing the practical performance of these algorithms.
CFs provide one route to achieve this goal, providing substantial variance reductions in many of the examples we considered.
Indeed, CFs allow us to use a sub-optimal QMC rule (e.g. as built into existing software packages) and yet, with minimal additional coding, obtain a QMC+CF algorithm that attains optimal convergence rates.
The focus on unknown smoothness $\alpha$ distinguishes our work from previous literature on the connection between integration and functional approximation, e.g. \cite{Bakhvalov,Heinrich}.

Functional approximation, and hence our QMC+CF methodology, has a computational cost associated with solution of a linear system. 
Whilst negligible in our experiments, this cost could be reduced if necessary using standard approximations and/or compactly supported kernels.
On the other hand, we note that QMC is often used when $f$ is expensive to evaluate and in such situations it is likely that evaluation of the integrand, rather than solution of a linear system, will be the main computational bottleneck.

Our focus was on Sobolev spaces, but it is known that a faster rate $O(N^{-\alpha + \epsilon})$ is possible in the subspace $SH^{\alpha}([0,1]^d)$, for any $\epsilon > 0$, and explicit point sets are available (for integer $\alpha$) \cite{Dick6}.
An immediate extension is to establish optimal rates for CFs in this class of functions.
In a related direction, one can in principle obtain {\it dimension-independent} rates by imposing a (strong) assumption of polynomial tractability on the RKHS. 
This is achieved by generalising to weighted Sobolev spaces, such that the integrand $f$ `depends only weakly on most of the components of $\bm{x}$'.
Further details are provided in \cite{Dick,Novak2} and form part of our ongoing research.

The methods that we describe are immediately applicable in a range of applications including marginalisation of hyper-parameters in classification \cite{Filippone}, probabilistic inference for differential equations \cite{Schober,Cockayne}, computation of model evidence \cite{Oates4} and approximation of the partition function in social network models \cite{Robins}.
Finally we note that CFs generalise to other integration methods including Bayesian Quadrature \cite{OHagan,Briol} and related kernel-based quadrature rules \cite{Bach}, in which the worst case error is also controlled by an RKHS norm $\|f\|_H$; this will be the focus of our ongoing research.

\subsubsection*{Acknowledgments}

The authors are grateful to Dan Simpson, Mathieu Gerber and Ben Collyer for helpful discussions.
CJO was supported by EPSRC [EP/D002060/1] and the ARC Centre of Excellence for Mathematics and Statistical Frontiers. 
MG was supported by EPSRC [EP/J016934/1, EP/K034154/1], an EPSRC Established Career Fellowship, the EU grant [EU/259348] and a Royal Society Wolfson Research Merit Award.

\subsubsection*{References}

\renewcommand{\section}[2]{}

\newpage\onecolumn

\subsection*{Supplementary Text}

In this section we provide details for how to allocate computational resources between the sets $\bm{u}^{1:M}$ and $\bm{v}^{M+1:N}$, trading off integration error with functional approximation error.

\begin{proposition*}
The optimal scaling of $M/N \in (0,1)$, in the sense of asymptotically minimising the QMC+CF absolute error, is given by
$$
\frac{M}{N} \rightarrow c^* = \frac{\alpha - \alpha_L}{\alpha}.
$$
\end{proposition*}
\begin{proof}
From Theorem \ref{theo1}, with $\bm{u}^{1:M}$ quasi-uniform, the QMC+CF error is bounded above by
\begin{eqnarray}
| E[f;\bm{u}^{1:M},\bm{v}^{M+1:N}] - I[f] |  \leq C_f e_{H^{\alpha_L}}(\bm{v}^{M+1:N}) h(\bm{u}^{1:M})^{\alpha - \alpha_L} \label{supp bound}
\end{eqnarray}
for some constant $C_f \in (0,\infty)$.

For $\alpha_L$-optimal QMC we have that $e_{H^{\alpha_L}}(\bm{v}^{M+1:N}) = O((N-M)^{-\alpha_L/d})$.

Suppose that $M = m^d$ for some $m \in \mathbb{N}$. Then since $\bm{u}^{1:M}$ are quasi-uniform it follows (from considering a regular square lattice) that $h(\bm{u}^{1:M}) = O(d^{1/2}m^{-1})$. This gives the general scaling $h(\bm{u}^{1:M}) = O(d^{-1/2}M^{-1/d})$.

Writing $M = c N$ for some $c$ we obtain from Eqn. \ref{supp bound} the objective function
$$
J(c) = (N-(cN))^{-\alpha_L/d} \times (d^{-1/2}(cN)^{-1/d})^{\alpha - \alpha_L}
$$
that we wish to minimise over $c \in [0,1)$.
Solving for $J'(c) = 0$ completes the argument.
\end{proof}

\subsection*{Supplementary Experimental Results}

This section contains all simulated data results discussed in the paper.
Specifically, for each of the 6 test functions described by \cite{Genz}, we display results based on Wendland's compactly supported regression kernel \cite{Wendland} with smoothness parameter $k$ (described above) set equal to either
\begin{itemize}
\item $k=0$, or
\item $k = 1$, or
\item $k = 2$
\end{itemize}
in combination with QMC design points $\bm{v}^{M+1:N}$ generated from either
\begin{itemize}
\item a Halton sequence, deterministically scrambled using the reverse radix algorithm \cite{Kocis}, and then applying a random shift or
\item a Sobol sequence, randomly scrambled using the algorithm of \cite{Matousek}.
\end{itemize}

We used the MATLAB implementation of \cite{Genz} that is freely available (web address given in the Main Text).
The QMC design points can be generated using the in-build MATLAB functions \verb+haltonset+, \verb+sobolset+ and \verb+scramble+.
Full MATLAB code used to generate these results is provided in the Electronic Supplement.

\newpage

\makeatletter 
\renewcommand{\thefigure}{S\@arabic\c@figure}
\makeatother

\newpage
\subsection*{Genz Function \#1: Oscillatory Test Function}

\begin{figure}[h!]
\centering
\begin{subfigure}[b]{0.42\textwidth}
\includegraphics[width = \textwidth]{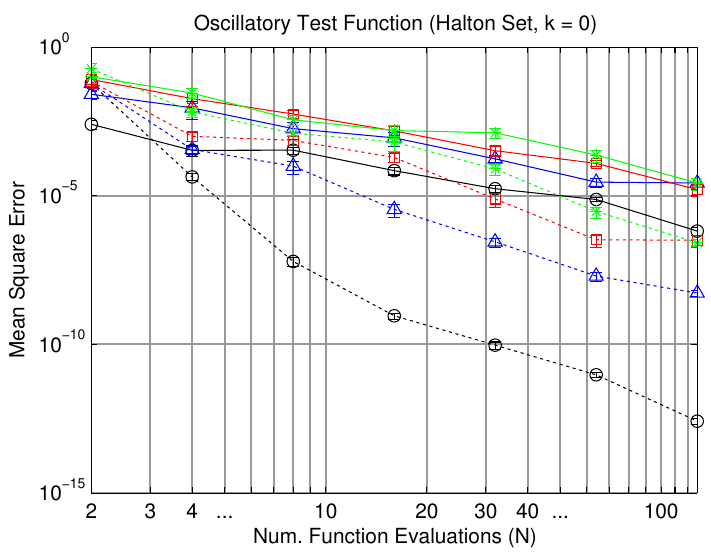}
\end{subfigure}
\begin{subfigure}[b]{0.42\textwidth}
\includegraphics[width = \textwidth]{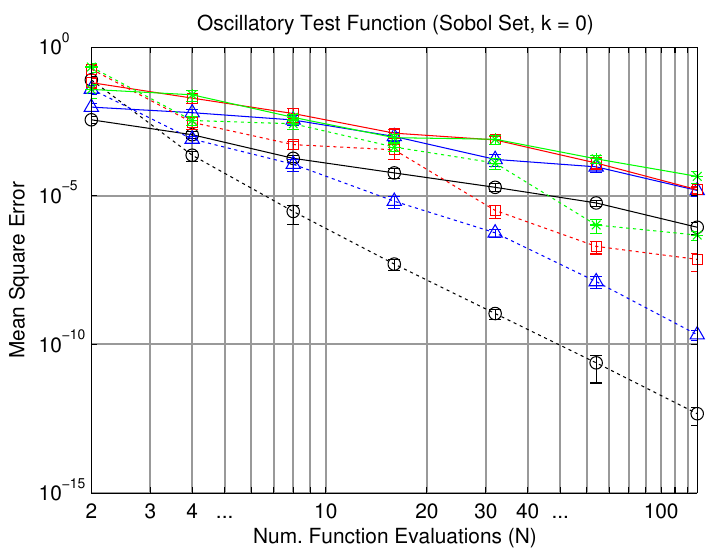}
\end{subfigure}

\begin{subfigure}[b]{0.42\textwidth}
\includegraphics[width = \textwidth]{Figures/func1_k1_h.pdf}
\end{subfigure}
\begin{subfigure}[b]{0.42\textwidth}
\includegraphics[width = \textwidth]{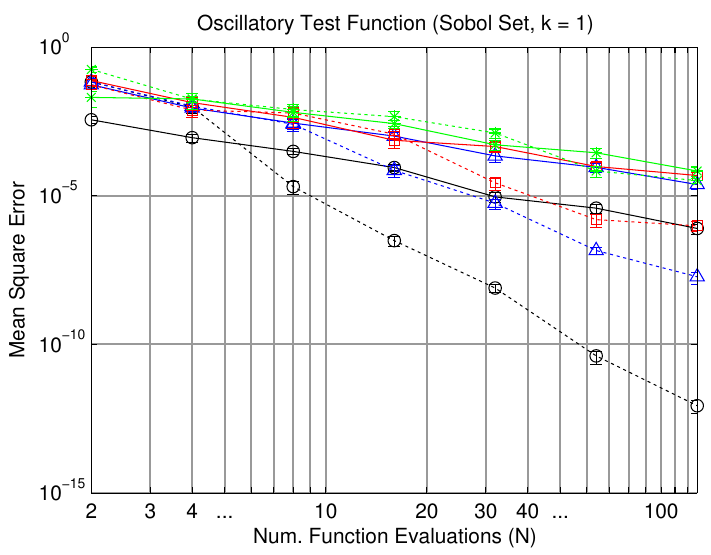}
\end{subfigure}

\begin{subfigure}[b]{0.42\textwidth}
\includegraphics[width = \textwidth]{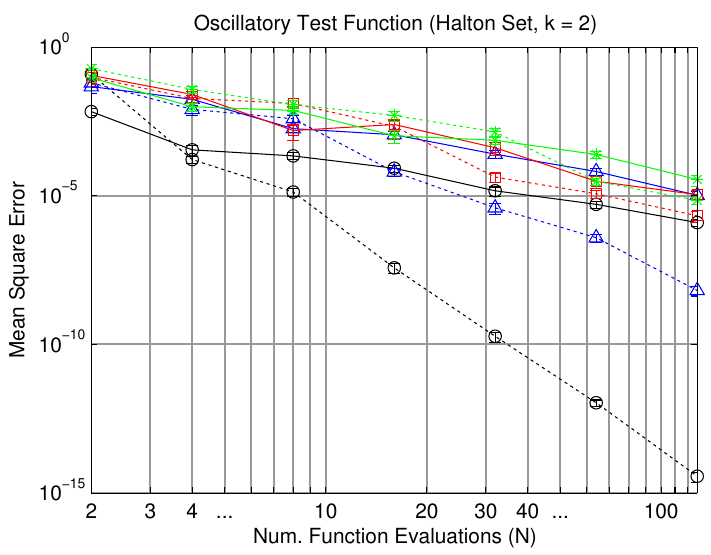}
\end{subfigure}
\begin{subfigure}[b]{0.42\textwidth}
\includegraphics[width = \textwidth]{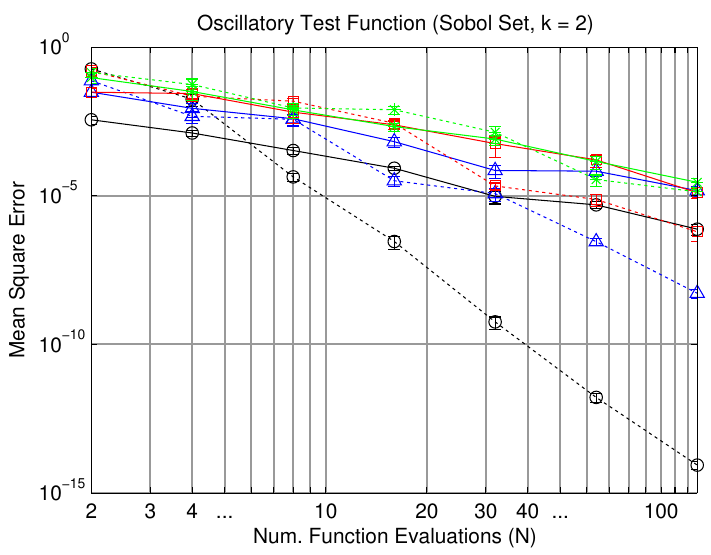}
\end{subfigure}
\caption{Numerical results: 
Each panel represents one of the 6 QMC+CF formulations. 
Solid lines correspond to standard QMC, dashed lines correspond to QMC+CF.
$\ocircle$ represents dimension $d=1$, \textcolor{blue}{$\triangle$} represents $d=2$, \textcolor{red}{$\square$} represents $d = 3$ and \textcolor{green}{$*$} represents $d = 4$.
Experiments were replicated with 10 random seeds and error bars denote standard error of the replicate mean.
QMC points were generated either from a scrambled Halton sequence or a scrambled Sobol sequence (see the Main Text).
The Wendland regression kernel took parameter $k$.
}
\end{figure}

\newpage
\subsection*{Genz Function \#2: Product Peak Test Function}

\begin{figure}[h!]
\centering
\begin{subfigure}[b]{0.42\textwidth}
\includegraphics[width = \textwidth]{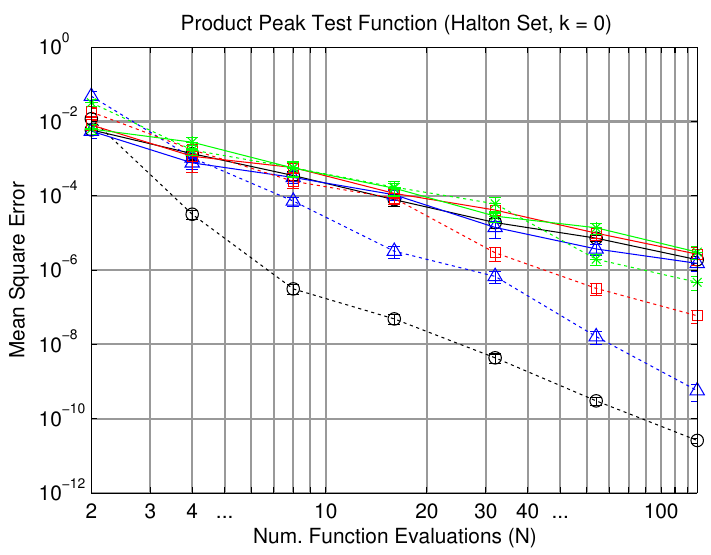}
\end{subfigure}
\begin{subfigure}[b]{0.42\textwidth}
\includegraphics[width = \textwidth]{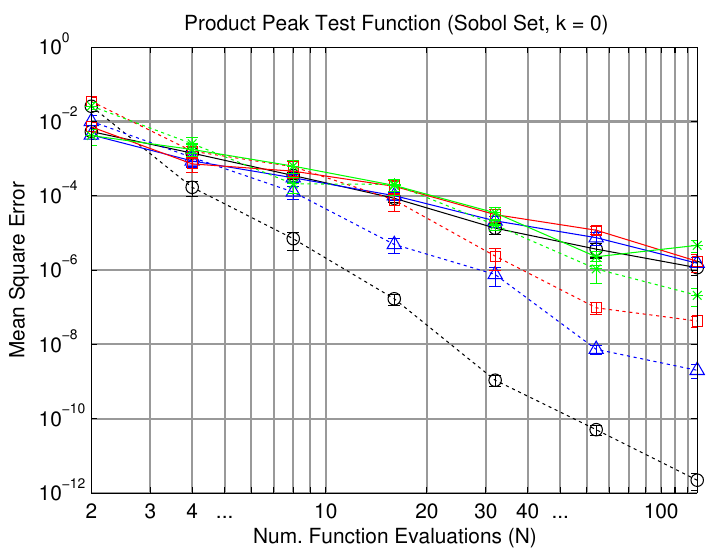}
\end{subfigure}

\begin{subfigure}[b]{0.42\textwidth}
\includegraphics[width = \textwidth]{Figures/func2_k1_h.pdf}
\end{subfigure}
\begin{subfigure}[b]{0.42\textwidth}
\includegraphics[width = \textwidth]{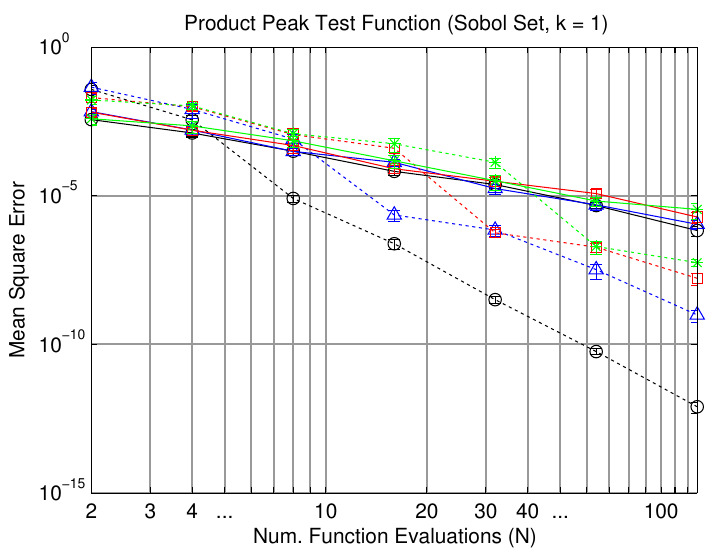}
\end{subfigure}

\begin{subfigure}[b]{0.42\textwidth}
\includegraphics[width = \textwidth]{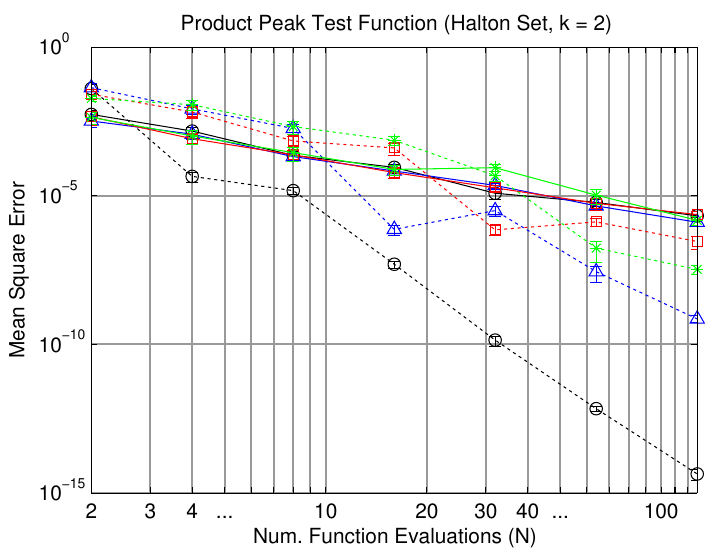}
\end{subfigure}
\begin{subfigure}[b]{0.42\textwidth}
\includegraphics[width = \textwidth]{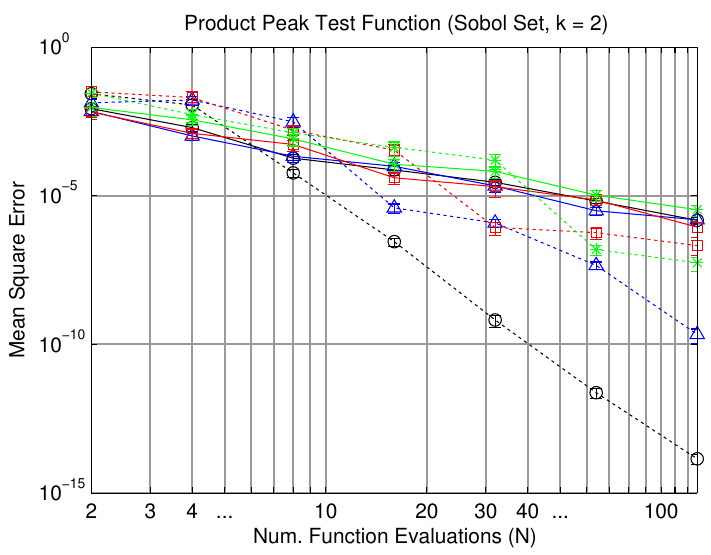}
\end{subfigure}
\caption{Numerical results: 
Each panel represents one of the 6 QMC+CF formulations. 
Solid lines correspond to standard QMC, dashed lines correspond to QMC+CF.
$\ocircle$ represents dimension $d=1$, \textcolor{blue}{$\triangle$} represents $d=2$, \textcolor{red}{$\square$} represents $d = 3$ and \textcolor{green}{$*$} represents $d = 4$.
Experiments were replicated with 10 random seeds and error bars denote standard error of the replicate mean.
QMC points were generated either from a scrambled Halton sequence or a scrambled Sobol sequence (see the Main Text).
The Wendland regression kernel took parameter $k$.}
\end{figure}

\newpage
\subsection*{Genz Function \#3: Corner Peak Test Function}

\begin{figure}[h!]
\centering
\begin{subfigure}[b]{0.42\textwidth}
\includegraphics[width = \textwidth]{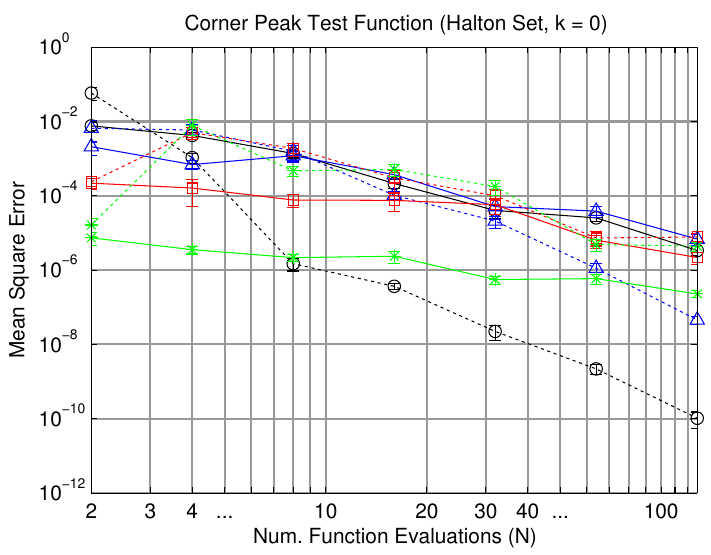}
\end{subfigure}
\begin{subfigure}[b]{0.42\textwidth}
\includegraphics[width = \textwidth]{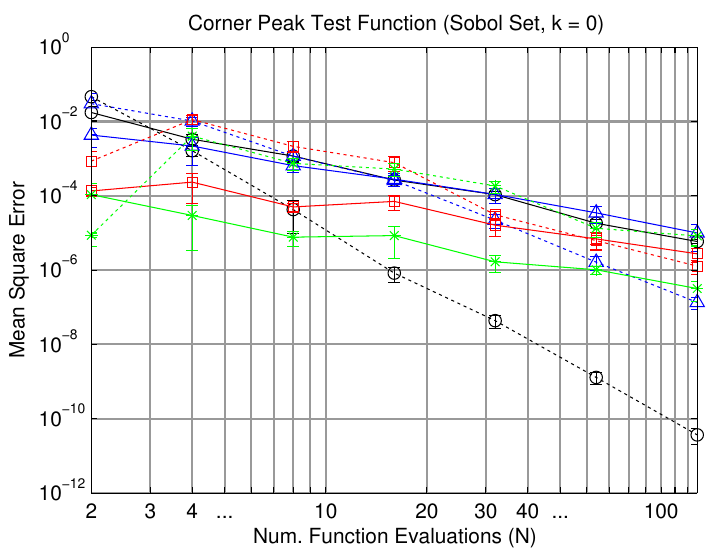}
\end{subfigure}

\begin{subfigure}[b]{0.42\textwidth}
\includegraphics[width = \textwidth]{Figures/func3_k1_h.pdf}
\end{subfigure}
\begin{subfigure}[b]{0.42\textwidth}
\includegraphics[width = \textwidth]{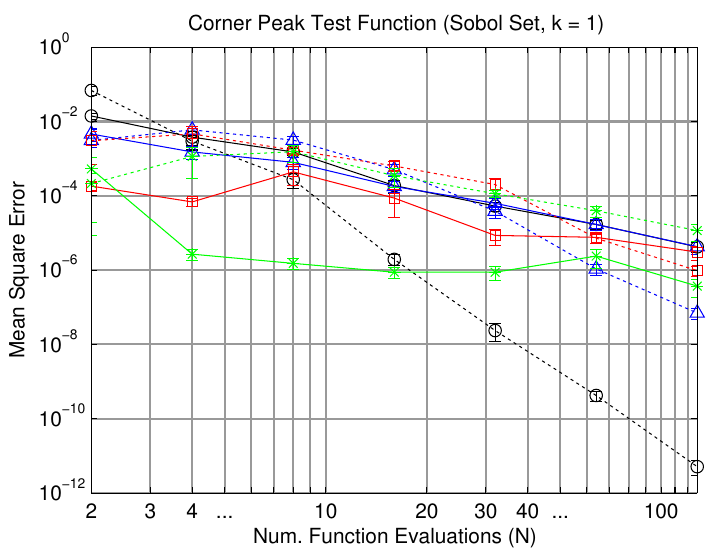}
\end{subfigure}

\begin{subfigure}[b]{0.42\textwidth}
\includegraphics[width = \textwidth]{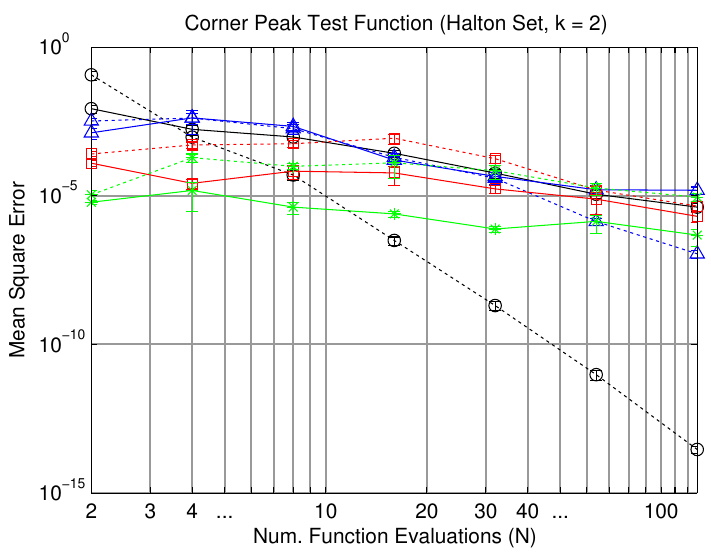}
\end{subfigure}
\begin{subfigure}[b]{0.42\textwidth}
\includegraphics[width = \textwidth]{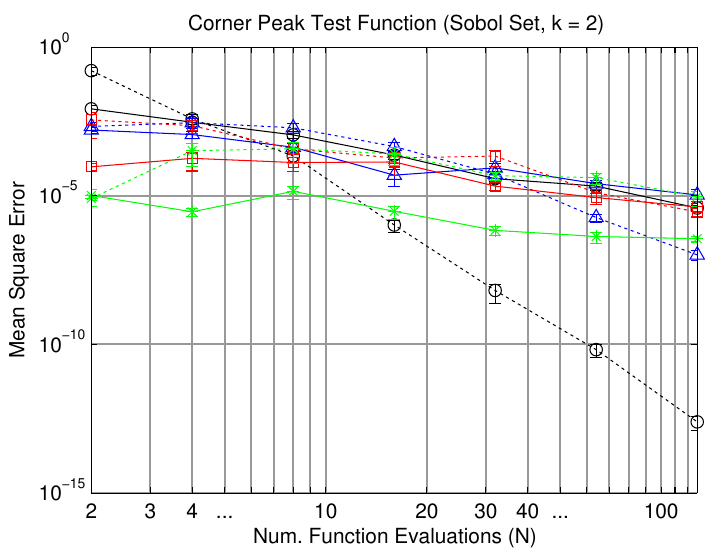}
\end{subfigure}
\caption{Numerical results: 
Each panel represents one of the 6 QMC+CF formulations. 
Solid lines correspond to standard QMC, dashed lines correspond to QMC+CF.
$\ocircle$ represents dimension $d=1$, \textcolor{blue}{$\triangle$} represents $d=2$, \textcolor{red}{$\square$} represents $d = 3$ and \textcolor{green}{$*$} represents $d = 4$.
Experiments were replicated with 10 random seeds and error bars denote standard error of the replicate mean.
QMC points were generated either from a scrambled Halton sequence or a scrambled Sobol sequence (see the Main Text).
The Wendland regression kernel took parameter $k$.}
\end{figure}

\newpage
\subsection*{Genz Function \#4: Gaussian Test Function}

\begin{figure}[h!]
\centering
\begin{subfigure}[b]{0.42\textwidth}
\includegraphics[width = \textwidth]{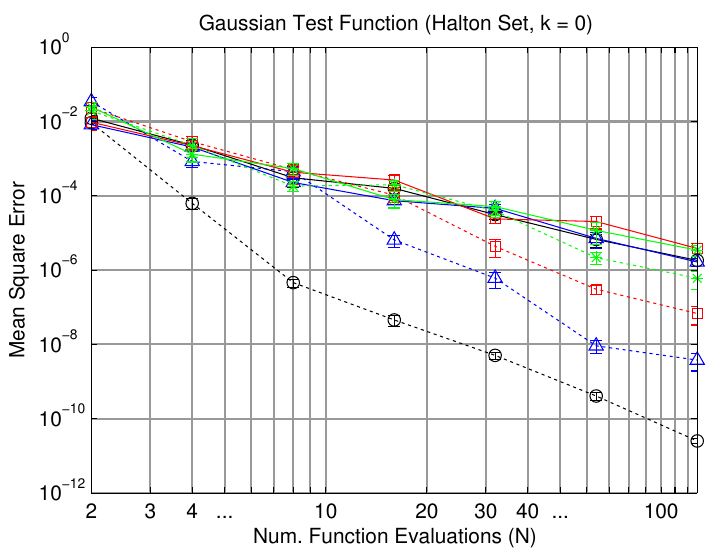}
\end{subfigure}
\begin{subfigure}[b]{0.42\textwidth}
\includegraphics[width = \textwidth]{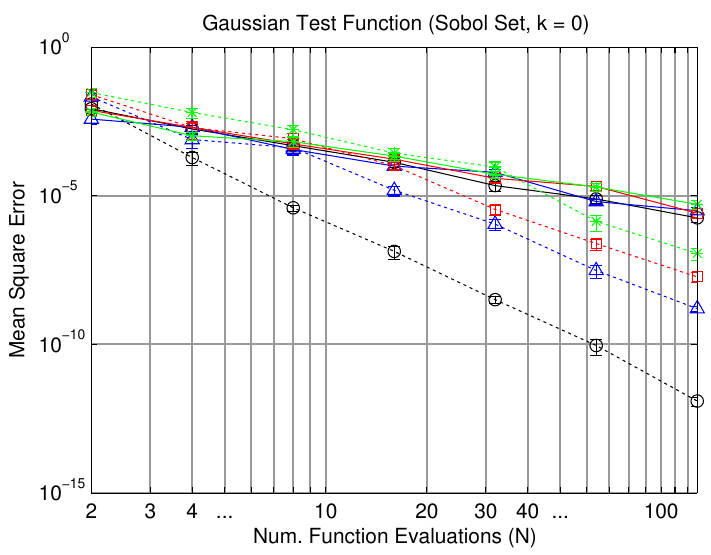}
\end{subfigure}

\begin{subfigure}[b]{0.42\textwidth}
\includegraphics[width = \textwidth]{Figures/func4_k1_h.pdf}
\end{subfigure}
\begin{subfigure}[b]{0.42\textwidth}
\includegraphics[width = \textwidth]{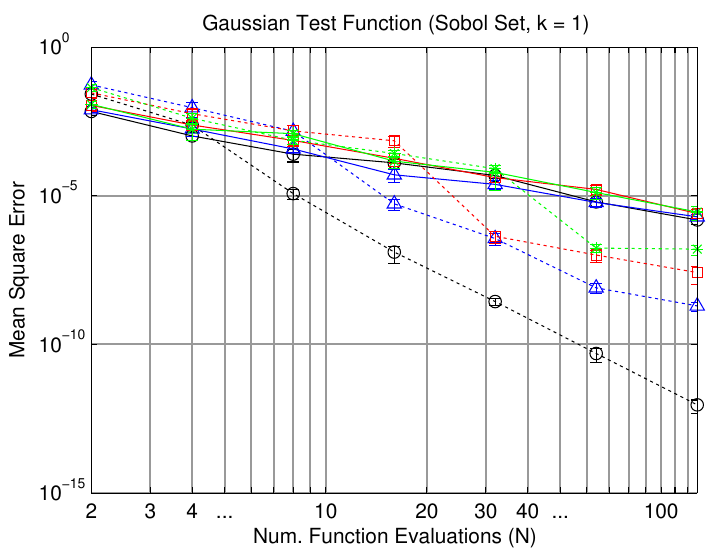}
\end{subfigure}

\begin{subfigure}[b]{0.42\textwidth}
\includegraphics[width = \textwidth]{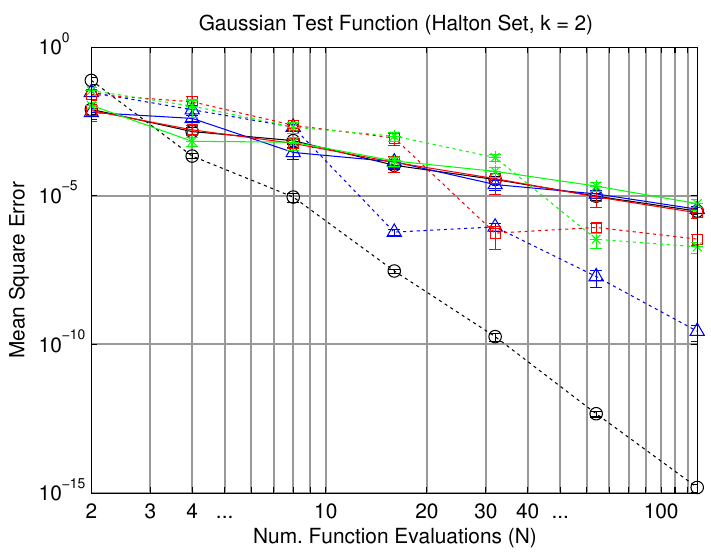}
\end{subfigure}
\begin{subfigure}[b]{0.42\textwidth}
\includegraphics[width = \textwidth]{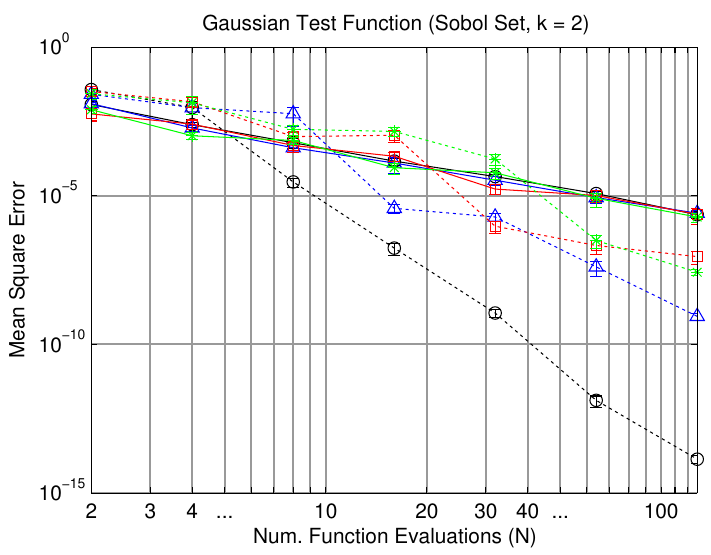}
\end{subfigure}
\caption{Numerical results: 
Each panel represents one of the 6 QMC+CF formulations. 
Solid lines correspond to standard QMC, dashed lines correspond to QMC+CF.
$\ocircle$ represents dimension $d=1$, \textcolor{blue}{$\triangle$} represents $d=2$, \textcolor{red}{$\square$} represents $d = 3$ and \textcolor{green}{$*$} represents $d = 4$.
Experiments were replicated with 10 random seeds and error bars denote standard error of the replicate mean.
QMC points were generated either from a scrambled Halton sequence or a scrambled Sobol sequence (see the Main Text).
The Wendland regression kernel took parameter $k$.}
\end{figure}

\newpage
\subsection*{Genz Function \#5: Continuous Test Function}

\begin{figure}[h!]
\centering
\begin{subfigure}[b]{0.42\textwidth}
\includegraphics[width = \textwidth]{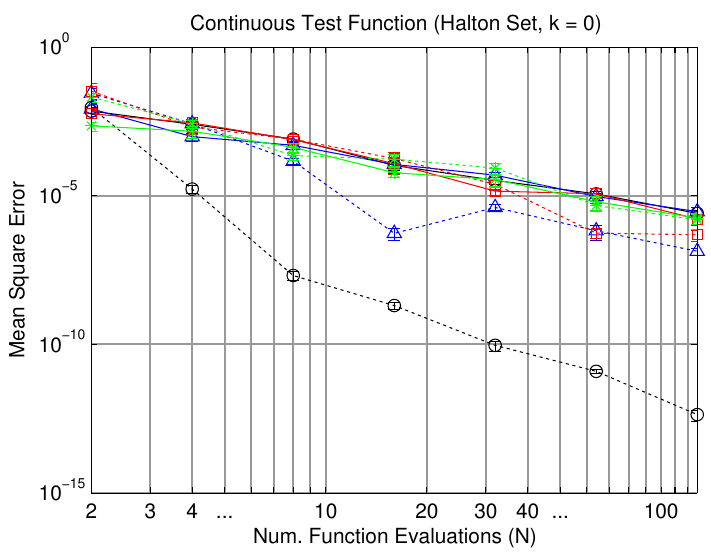}
\end{subfigure}
\begin{subfigure}[b]{0.42\textwidth}
\includegraphics[width = \textwidth]{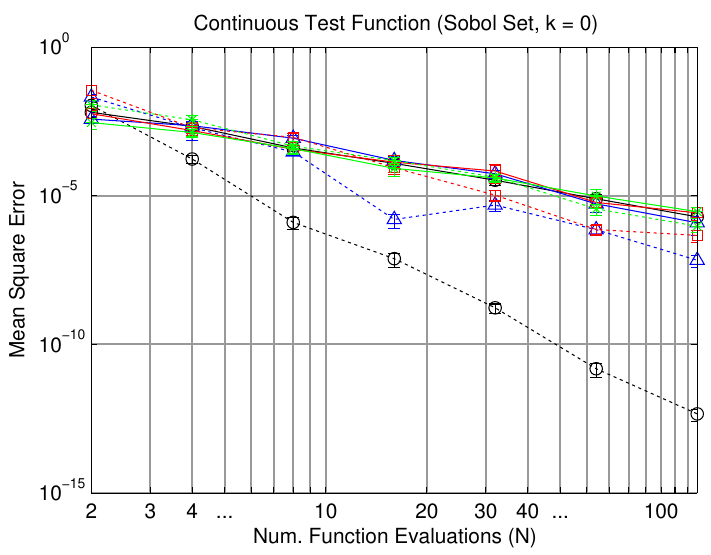}
\end{subfigure}

\begin{subfigure}[b]{0.42\textwidth}
\includegraphics[width = \textwidth]{Figures/func5_k1_h.pdf}
\end{subfigure}
\begin{subfigure}[b]{0.42\textwidth}
\includegraphics[width = \textwidth]{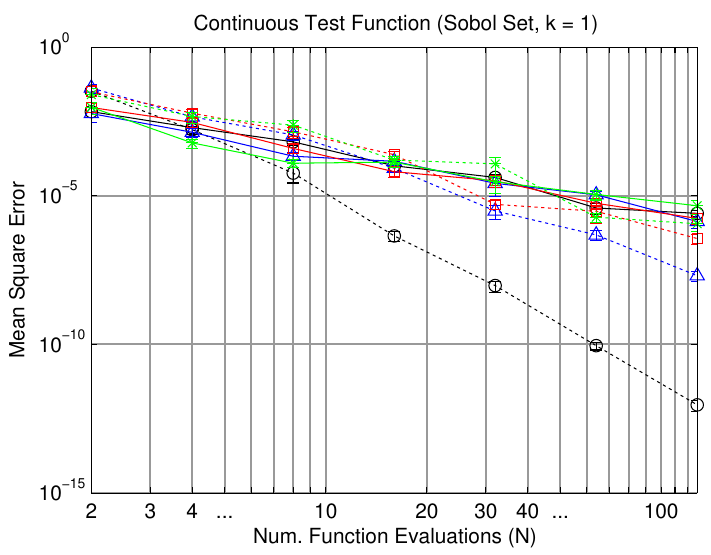}
\end{subfigure}

\begin{subfigure}[b]{0.42\textwidth}
\includegraphics[width = \textwidth]{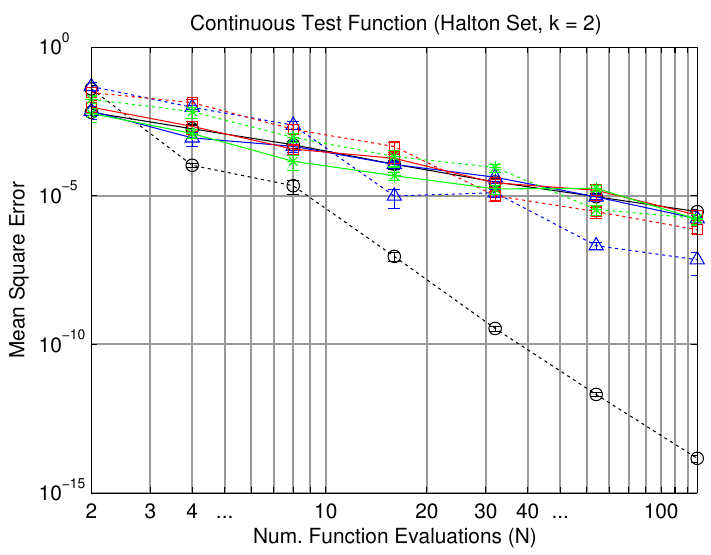}
\end{subfigure}
\begin{subfigure}[b]{0.42\textwidth}
\includegraphics[width = \textwidth]{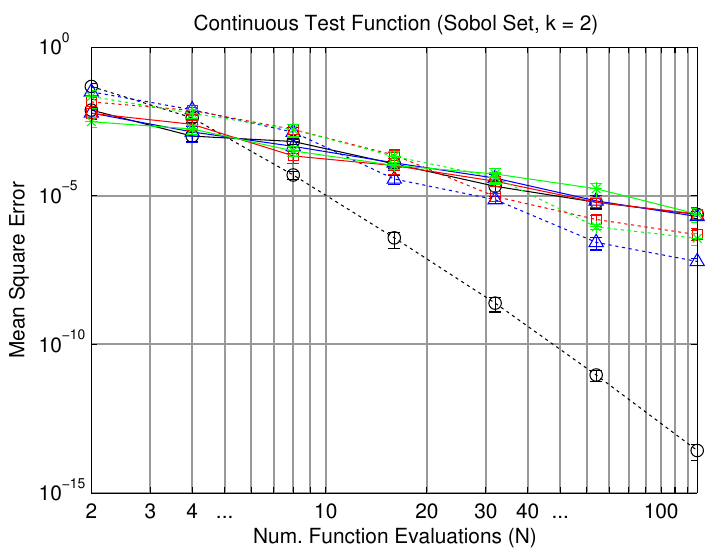}
\end{subfigure}
\caption{Numerical results: 
Each panel represents one of the 6 QMC+CF formulations. 
Solid lines correspond to standard QMC, dashed lines correspond to QMC+CF.
$\ocircle$ represents dimension $d=1$, \textcolor{blue}{$\triangle$} represents $d=2$, \textcolor{red}{$\square$} represents $d = 3$ and \textcolor{green}{$*$} represents $d = 4$.
Experiments were replicated with 10 random seeds and error bars denote standard error of the replicate mean.
QMC points were generated either from a scrambled Halton sequence or a scrambled Sobol sequence (see the Main Text).
The Wendland regression kernel took parameter $k$.}
\end{figure}

\newpage
\subsection*{Genz Function \#6: Discontinuous Test Function}

\begin{figure}[h!]
\centering
\begin{subfigure}[b]{0.42\textwidth}
\includegraphics[width = \textwidth]{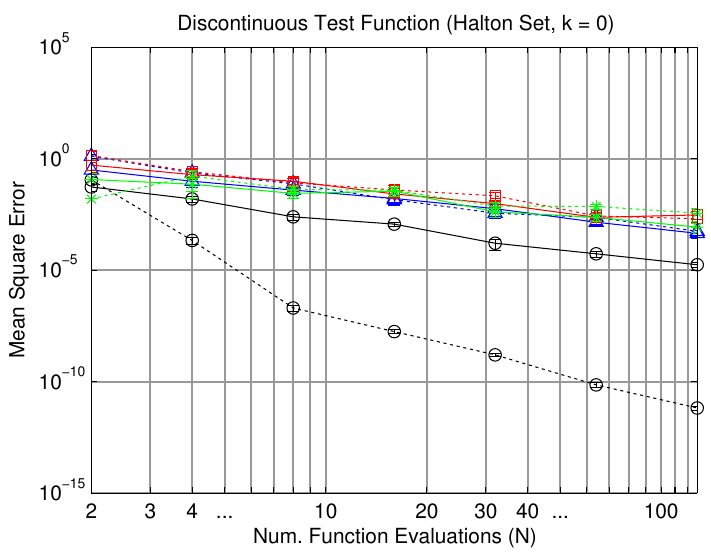}
\end{subfigure}
\begin{subfigure}[b]{0.42\textwidth}
\includegraphics[width = \textwidth]{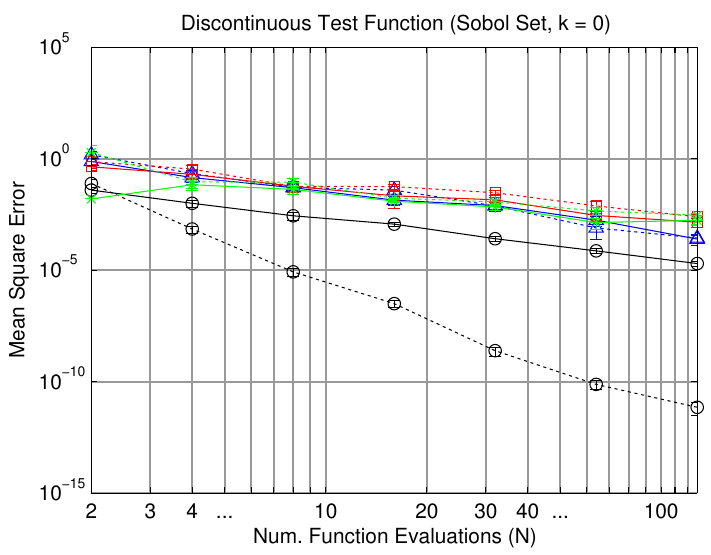}
\end{subfigure}

\begin{subfigure}[b]{0.42\textwidth}
\includegraphics[width = \textwidth]{Figures/func6_k1_h.pdf}
\end{subfigure}
\begin{subfigure}[b]{0.42\textwidth}
\includegraphics[width = \textwidth]{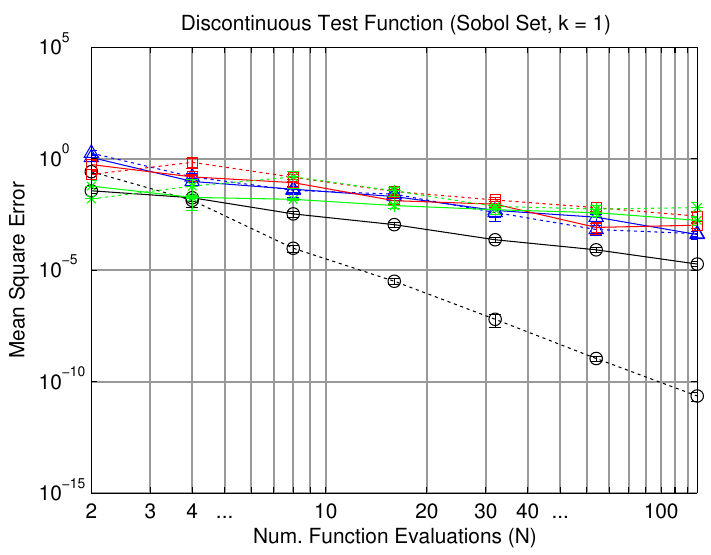}
\end{subfigure}

\begin{subfigure}[b]{0.42\textwidth}
\includegraphics[width = \textwidth]{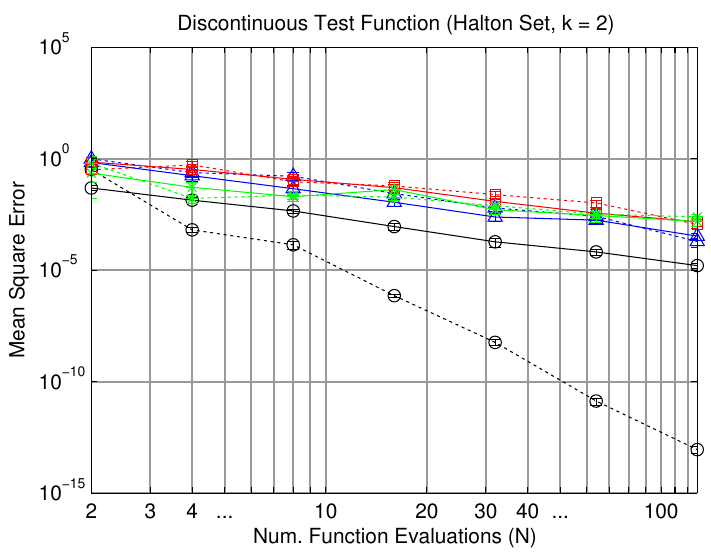}
\end{subfigure}
\begin{subfigure}[b]{0.42\textwidth}
\includegraphics[width = \textwidth]{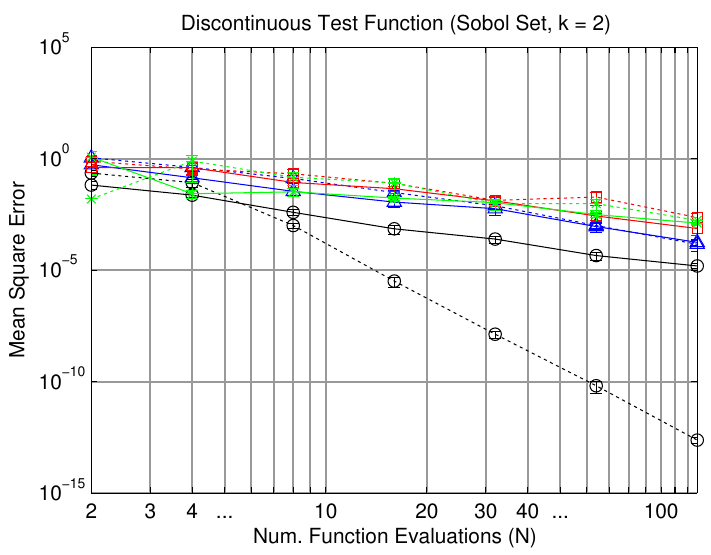}
\end{subfigure}
\caption{Numerical results: 
Each panel represents one of the 6 QMC+CF formulations. 
Solid lines correspond to standard QMC, dashed lines correspond to QMC+CF.
$\ocircle$ represents dimension $d=1$, \textcolor{blue}{$\triangle$} represents $d=2$, \textcolor{red}{$\square$} represents $d = 3$ and \textcolor{green}{$*$} represents $d = 4$.
Experiments were replicated with 10 random seeds and error bars denote standard error of the replicate mean.
QMC points were generated either from a scrambled Halton sequence or a scrambled Sobol sequence (see the Main Text).
The Wendland regression kernel took parameter $k$.}
\end{figure}

\newpage
\subsection*{Robot Arm Example: Additional Results}

We re-ran the robot arm simulation in order to compare the QMC+CF estimator with the MC+CF estimator; that is, a quasi-uniform set $\bm{u}^{1:M}$ were used to construct a control functional $f_M$, whilst a Monte Carlo sample $\bm{v}^{M+1:N}$ were used to integrate the difference $f - f_M$.

\begin{figure}[h!]
\centering
\includegraphics[width = 0.5\textwidth,clip,trim = 6.5cm 10cm 3cm 10cm]{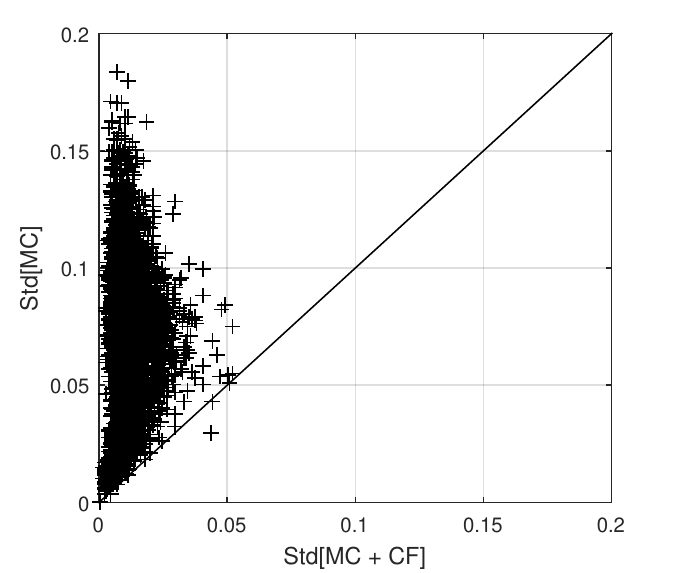}
\caption{Application to robot arm data: Examining the estimator sampling standard deviations, we see that, for all but a handful of the configurations, QMC+CF was more accurate than MC+CF.}
\end{figure}

\end{document}